\documentclass[pdflatex,sn-mathphys-num]{sn-jnl}

\usepackage{amsmath, amssymb, amsfonts}
\usepackage{mathrsfs}
\usepackage{mathtools}
\usepackage{graphicx}
\usepackage{multirow}
\usepackage{booktabs}
\usepackage{manyfoot}
\usepackage{appendix}
\usepackage{xcolor}
\usepackage{textcomp}
\usepackage{upgreek}
\usepackage{xfrac}
\usepackage{marginnote}
\usepackage{hyperref} 
\usepackage[mathscr]{euscript}
\usepackage{cleveref}
\usepackage{adjustbox}
\usepackage{stmaryrd}
\usepackage[T1]{fontenc}
\usepackage[utf8]{inputenc}
\usepackage[scaled]{inconsolata}
\usepackage{multicol}
\usepackage{longtable}
\usepackage{MnSymbol}
\usepackage{comment}
\usepackage{soul}
\usepackage{pdflscape}
\usepackage[font=footnotesize, tableposition=top]{caption}
\usepackage{array}
\usepackage{bbm}
\usepackage{refcount}
\usepackage{tikz}
\usetikzlibrary{cd, external, decorations.pathmorphing, fit, backgrounds, tikzmark, calc, shapes.misc, positioning}
\usepackage{tikz-cd}
\usepackage{algorithm}
\usepackage{algorithmicx}
\usepackage{algpseudocode}
\usepackage{listings}
\usepackage{lstautogobble}
\usepackage{lstlinebgrd}
\usepackage{helvet}
\usepackage{courier}
\usepackage{type1cm}
\usepackage{makeidx}
\usepackage[bottom]{footmisc}
\usepackage[most]{tcolorbox}
\tcbuselibrary{theorems, skins, breakable}

\definecolor{keywordcolor}{rgb}{0.7, 0.1, 0.1}
\definecolor{tacticcolor}{rgb}{0.0, 0.1, 0.6}
\definecolor{commentcolor}{rgb}{0.4, 0.4, 0.4}
\definecolor{symbolcolor}{rgb}{0.0, 0.1, 0.6}
\definecolor{sortcolor}{rgb}{0.1, 0.5, 0.1}
\definecolor{attributecolor}{rgb}{0.7, 0.1, 0.1}

\tcbset{
  minimaliststyle/.style={
    enhanced,
    breakable,
    colback=white,
    colframe=darkgray!10,
    fonttitle=\bfseries,
    coltitle=black,
    rounded corners,
    boxed title style={
      colback=gray!10,
      colframe=gray!50,
      coltext=black,
      rounded corners,
    },
    before skip=8pt,
    after skip=8pt,
  }
}

\newcommand{\mathcolorbox}[2]{\colorbox{#1}{$#2$}}

\newcommand{\Lean}{\textsf{Lean4}}

\newcommand{\NN}{\ensuremath{\mathbb{N}}}
\newcommand{\ZZ}{\ensuremath{\mathbb{Z}}}

\newcommand{\ox}{\ensuremath{\otimes}}

\newcommand{\McA}{\mathcal{A}}

\newcommand{\mcA}{\mathcal{A}}
\newcommand{\mcB}{\mathcal{B}}

\newcommand{\mcM}{\mathcal{M}}

\newcommand{\mfp}{\ensuremath{\mathfrak{p}}}

\DeclareMathOperator{\im}{\ensuremath{im}}

\newcommand{\id}{\ensuremath{\mathsf{id}}}

\newcommand{\quotient}[2]{\left.\raisebox{0.8ex}{$#1$}\!\middle/\!\raisebox{-0.8ex}{$#2$}\right.}

\newcommand{\decompose}{\ensuremath{\mathtt{decomp}}}
\newcommand{\recompose}{\ensuremath{\mathtt{recomp}}}
\newcommand{\hloc}[2]{{#1}_{\left(#2\right)}}
\newcommand{\hloca}[2]{\hloc{#1}{#2}}

\DeclareMathOperator{\proj}{\ensuremath{Proj}}

\DeclareMathOperator{\spec}{\ensuremath{Spec}}
\newcommand{\rad}[1]{\sqrt{#1}}
\DeclareMathOperator{\potion}{\ensuremath{\vartheta}}
\newcommand{\leanmath}[1]{\ensuremath{\footnotesize{\mathtt{#1}}}}
\newcommand{\sfs}{\ensuremath{\mathsf{s}}}
\newcommand{\overbar}[1]{\mkern 1.6mu\overline{\mkern-1.6mu#1\mkern-1.6mu}\mkern 1.6mu}
\newcommand{\mybar}[1]{\ensuremath{\overbar{#1}}}

\newcommand{\scrF}{\ensuremath{\mathscr{F}}}

\newtheorem{theorem}{Theorem}
\newtheorem{proposition}[theorem]{Proposition}
\newtheorem{lemma}[theorem]{Lemma}
\newtheorem{corollary}[theorem]{Corollary}
\newtheorem{example}{Example}
\newtheorem{remark}{Remark}
\newtheorem{definition}{Definition}
\newtheorem{construction}{Construction}

\title{Formalizing multi-graded Brenner-Schröer Proj schemes and dilatations of rings in Lean4}

\author*[3,4]{\fnm{Arnaud} \sur{Mayeux}}\email{mayeux@wisc.edu}
\author[1,2]{\fnm{Jujian} \sur{Zhang}}

\affil[1]{\orgname{Axiom Math},  \country{USA}}

\affil[2]{\orgname{Imperial College London},  \country{UK}}

 \affil[3] {\orgname{The Hebrew University of Jerusalem}, \country{IL}}

\affil[4] {\orgname{University of Wisconsin-Madison}, \country{USA}}

\abstract{
We present a formalization in Lean4 of some multi-graded algebraic geometry constructions, focusing on the Brenner--Schröer Proj construction and algebraic dilatations of rings. Multi-graded Proj schemes, defined from rings graded by more general monoids than $\mathbb{N}$ or $\mathbb{Z}$, have recently attracted increasing attention and play an important role in several areas of modern algebraic geometry.

Our work follows the algebraic approach developed in the literature and provides a formal implementation of multi-graded Proj within the Lean4 theorem prover. In addition, we formalize dilatations of rings, an operation in commutative algebra closely related to localization and to blowup constructions.

This article gives a comprehensive account of the definitions, main results, and design choices underlying the formalization. It is intended both as documentation of the development and as a foundation for future extensions in formalized algebraic geometry. The corresponding code is made publicly available, supporting further developments in the formalization of advanced geometric structures.
}


\begin{document}
\maketitle

\tableofcontents

\section{Introduction}
\subsection{Overview and objectives of the formalization}
Multi-graded algebraic geometry refers to constructions in algebraic geometry built on rings graded by monoids or groups more general than $\mathbb{N}$ or $\mathbb{Z}$. In this work, we implement such constructions in Lean 4 \cite{Mou21}. The main focus of this work is the Brenner--Schröer Proj construction, introduced by Brenner and Schröer in \cite{BS03}. This construction has recently been studied in several works \cite{KSU21, MR24} and is poised to become a classical topic in algebraic geometry. Multi-graded Proj schemes are used in several areas of contemporary algebraic geometry, including Lie theory and geometric representation theory \cite[§2.3.7]{BV25} as well as the minimal model program \cite{Ya26}. If $A$ is a ring graded by the natural numbers and regarded as graded by the integers, then Grothendieck's $\operatorname{Proj}(A)$ \cite{Gr61} and Brenner--Schröer's $\operatorname{Proj}(A)$ \cite{BS03} are canonically identified \cite{BS03}.
In this work, we provide an implementation of the Brenner--Schröer construction of multi-graded Proj schemes within the Lean 4 theorem prover \cite{Mou21}, following the algebraic treatment in Mayeux-Riche \cite{MR24}.
We also formalize dilatations of rings, a fundamental operation in commutative algebra analogous to localization (cf. \cite{stacks-project, M24, DMdS23}) and related to blowups and part of the broader area of multi-graded algebraic geometry. Dilatations form a process that adds fractions to rings, refining localizations. They allow one to control that both the numerator and denominator are prescribed \cite{M24}.

Part of this formalization was announced in the short note \cite{MZ} and already discussed in the thesis of Jujian Zhang. The code is available on GitHub \cite{projconstruction}. 
 This article presents a complete account of the formalization and serves as a guide for future work building on these fundamental constructions.
 \subsection{Relation to other works}
 We now discuss the methods and related works used in this paper. 
The Grothendieck $\mathbb{N}$-graded Proj construction \cite{Gr61} was formalized in Lean by Zhang in \cite{Z23} and is now in mathlib. The work \cite{Z23} provided the first example of non-affine schemes ever formalized in a theorem prover.

In contrast to \cite{Gr61, Z23}, where Proj is first defined as a set of prime ideals, in the present work we define multi-graded Proj directly by gluing, following the approach of \cite{MR24} (itself a minor variation of the approach of \cite{BS03}), without passing through a set-theoretic construction involving graded ideals. 

The present development also relies on and is closely inspired by the pioneering formalization of graded algebra in Lean carried out by Zhang and Wieser in \cite{WZ22}.

Of course, our work builds on the implementation of the foundations of the theory of schemes in Lean \cite{Bual}, which is now part of mathlib.
\subsection{Mathematical description of multi-graded Proj schemes}
We now describe the technical core of the mathematical content formalized in this paper. We also formalize related results, but the main achievement of this formalization effort concerns the following material.

We recommend that readers familiarize themselves with the structure of this section before proceeding with the remainder of the paper, as it provides a roadmap for the material developed in subsequent sections.

Let $A$ be a commutative unital ring graded by a finitely generated abelian group $\iota$.

A (multiplicative) submonoid $S$ of $A$ is homogeneous if every element of $S$ is homogeneous. Equivalently a submonoid of $A$ is homogeneous if it is generated by homogeneous elements. 
In this situation, the localization $A_S$ of $A$ with respect to $S$ is canonically $\iota$-graded.
Given a homogeneous (multiplicative) submonoid $S \subset A$, we will denote by $\overline{S}$ the homogeneous submonoid consisting of homogeneous divisors of elements in $S$. Note that we have a canonical isomorphism of graded rings
 $A_{ \overline{S}} \cong A_{S}.
$

\begin{definition}[\cite{BS03,MR24}] \label{defi:intro1}
 Let $S$ be a homogeneous subset of $A$. We denote by $\deg(S)$ the subset of $M$ defined as 
$\deg(S):= \{ m \in \iota : \exists s \in S , s  \in A_m \}. $
\end{definition}
 
Note that $\deg(\{0\})=\iota$.
More generally, if $S$ is a homogeneous submonoid of $A$, then $\deg(S)$ is a submonoid of $\iota$.

\begin{definition}[\cite{BS03,MR24}] \label{defi:intro2}
Let $S$ be a homogeneous (multiplicative) submonoid of $A$. 
 We put $\iota[S]=\iota[S \rangle^{\mathrm{gp}}= M[S \rangle- M[S \rangle$, the subgroup of $M$ generated by $\deg(S)$.
\end{definition}

Formalizations of Definitions \ref{defi:intro1} and \ref{defi:intro2} are documented below in Definition \ref{def:deg}.

\begin{definition}[\cite{BS03,MR24}] \label{defi:intro3}
 A homogeneous submonoid $S$ of $A$ is called \emph{$\iota$-relevant} (or just \emph{relevant} if $\iota$ is clear from the context) if for any $m$ in $\iota$ there exists $n \in \mathbf{Z}_{> 0}$ such that $nm$ belongs to $\iota[\overline{S}]$, i.e. if $\iota/(\iota[\overline{S}])$ is a torsion abelian group.
\end{definition}

Formalization of Definition \ref{defi:intro3} is documented below in Definition \ref{def:relev} (Definition \ref{def:relev} also formalizes related notions of mathematical interest used in \cite{BS03, MR24}).

\begin{example} For example, let $A= \mathbb{Z}[X,Y]$ graded by $\mathbb{Z}^2$ with $\deg(X)=(1,0)$ and $\deg(Y)=(0,1)$. The homogeneous submonoid $S$ generated by $X$ in $A$ is not relevant because $\iota[\overline{S}]= \mathbb{Z} \times 0 \subset \mathbb{Z}^2$. The homogeneous submonoid $T$ generated by $XY$ in $A$ is relevant because $\iota[\overline{T}]= \mathbb{Z}^2$. \end{example}

Let $S$ be a homogeneous submonoid of $A$. 
The degree-$0$ part $(A_S)_0$ of the localization $A_{S}$ is denoted $A_{(S)}$ and is called the \emph{potion} of $A$ with respect to $S$. 
We have a canonical identification
$A_{(\overline{S})} \cong A_{(S)}$.
If $S$ and $T$ are submonoids of $A$, we will denote by $ST$ the submonoid of $A$ generated by $S \cup T$, i.e. $ST = \{st : s \in S, \, t \in T \}$. Of course, $ST$ is homogeneous if $S$ and $T$ are.
The following is the key result that makes the Proj construction work.

\begin{proposition}[\cite{MR24}] \label{prop:magicintro}
Let $S$ and $T$ be homogeneous finitely generated submonoids of $A$. 
\begin{enumerate}
\item 
\label{it:magic-10}
We have a canonical homomorphism of potion rings 
$A_{(S )} \to  A_{(ST )}$.
\item 
\label{it:magic-30}
Assume that $S$ is relevant. The morphism of schemes 
$
\mathrm{Spec} (A_{(ST )}) \to \mathrm{Spec} (A_{(S)})
$
induced by the ring homomorphism in (1) is an open immersion of schemes.
\end{enumerate}
\end{proposition}

The proof of Proposition \ref{prop:magicintro}$\,$(\ref{it:magic-30}) consists in showing that $A_{(ST)}$ can be identified with the localisation of $A_{(S)}$ with respect to a certain finitely generated submonoid. Proposition \ref{prop:magicintro} corresponds to Theorem \ref{potion-loc-iso} (and preliminary material) below.

We denote by $\mathcal{F}_A$ the set of all relevant homogeneous submonoids of $A$ which are finitely generated as submonoids of $(A,\times)$.

\begin{definition}[\cite{BS03,MR24}] \label{defi:intropro}
Let
$\mathcal{F} \subset \mathcal{F}_A$ be a subset.
 For each $S \in \mathcal{F}$, let $D_{\dagger}(S)$ be the spectrum of the potion $A_{(S)}$. 
 If $S,T \in \mathcal{F}$, the affine scheme $D_{\dagger}(ST)$ identifies canonically with an open subscheme of $D_{\dagger}(S)$. For each $S,T \in \mathcal{F}$, we have equalities
$
 D_{\dagger}({S S}) = D_{\dagger}(S) $ and $ D_{\dagger}({ST})=D_{\dagger}({TS}).
$
 Moreover, for each triple $S,T,U \in \mathcal{F}$, we have
$
 D_{\dagger}({ST} )\cap D_{\dagger}({SU}) = D_{\dagger}({TS}) \cap D_{\dagger}({TU}).
$
Now, by glueing, from these data we obtain a scheme $\mathrm{Proj}_{\mathcal{F} } (A)$ and, for each $S \in \mathcal{F}$, an open immersion $\varphi_S : D_{\dagger}(S) \to \mathrm{Proj}_{\mathcal{F}} (A)$, such that
$
\mathrm{Proj}_{\mathcal{F} } (A) = \bigcup_{S \in \mathcal{F} } \varphi_S( D_{\dagger}(S)).
$
In practice, we will often identify $D_{\dagger}(S) $ and $\varphi_S (D_{\dagger}(S))$. 
In the case when $\mathcal{F} = \mathcal{F}_A$,
 the scheme $\mathrm{Proj}_{\mathcal{F}_A} (A)$ is the Brenner-Schröer multi-graded Proj schemes.
\end{definition}

The formalization of Definition \ref{defi:intropro} is documented below in Construction \ref{const:proj} and Definition \ref{defi:full}.

\subsection{Relation to other multi-graded constructions and limit of the present work}

Recall that multi-graded algebraic geometry refers to constructions in algebraic geometry involving rings graded by commutative monoids or groups more general than $\mathbb{N}$ or $\mathbb{Z}$.

The multi-graded Proj construction also makes sense for multi-graded quasi-coherent algebras; see \cite{MR24}. Blowups of schemes (possibly multi-centered, cf. \cite{MR24, Ta26}) are defined via Proj of multi-graded quasi-coherent algebras. To define the multi-graded Proj of a quasi-coherent algebra, one may glue the Proj constructions of graded rings locally. However, as a limitation of the present work, we do not carry out this global case. We plan to formalize it in a future project, but doing so would require developing a substantial amount of foundational material on schemes in mathlib, going beyond the formalization of multi-graded algebraic geometry.  However, this is not out of reach. We believe that, within a few years, the methods developed in the present paper, combined with new material on quasi-coherent sheaves in mathlib, will make it possible to formalize global multi-graded Proj constructions. This limitation provides yet another motivation for pursuing the formalization of the Stacks Project \cite{stacks-project}.

Another closely related notion is that of affine blowups, also called dilatations. Dilatations of schemes are likewise defined locally via dilatations of rings. They are open subschemes of projective blowups \cite{M24, MR24}. In the present work, we formalize dilatations of rings.

 Although the formalization of Proj presented in this paper is very general and involves non-affine schemes, it can still be viewed as a local construction (its input is a ring, rather than a scheme or a quasi-coherent algebra) underlying a more general Proj construction. Similarly, dilatations of rings or affine schemes can be regarded as local special (and necessary) cases of the more general global construction of dilatations of schemes. In future independent works, one may envision a complete formalization of multi-graded algebraic geometry in Lean.

\subsection{Content}

Section \ref{sec:general-setup-graded} is about the general setup we use throughout the paper about graded rings. In this work, a graded ring is a commutative unital ring $A$ together with a grading $A= \bigoplus_{i \in \iota} A_i$ where $\iota$ is an abelian monoid $(A_i \cdot A_j \subset A_{i+j})$. From now on, we assume that $\iota$ is an abelian group in this introduction.  A multi-plicative submonoid of a graded ring is called homogeneous if its elements are homogeneous. An homogeneous submonoid $S$ is called relevant if the group generated by the degrees of the homogeneous divisors of $S$ is a torsion subgroup of $\iota$. Section \ref{sec:relevant-homogeneous-submonoid} is about homogeneous and relevant submonoids. In sections \ref{sec:homogeneous-localization} and \ref{sec:graded-loc}, we study localization of graded rings. Section \ref{sec:graded-tensor} is devoted to tensor product of graded rings and related technical results. Section \ref{sec:potions} formalizes the definition of potions. Potions are rings defined as degree zero part of homogeneous localizations. They are at the heart of the Proj construction. We establish many results on potions in sections \ref{sec:potions}, \ref{sec:glu-proj} and \ref{sec:enlarg-fam-potion}, these are preliminary to define Proj schemes. The definition of Brenner-Schröer Proj schemes is in section \ref{sec:glu-proj}. The functoriality of the Proj construction is established in section \ref{sec:functo}.  Section \ref{sec:dil-ring} formalizes the definition and the universal property of dilatations of rings. 

This paper also contains the formalization of several facts on multi-graded algebra, not strictly required to define the definition of the multi-graded Proj schemes, but are part of the theory of multi-graded Proj schemes. In particular, we formalized material on tensor products of graded rings, lemmas on relevant elements, and the ideal generated by relevant elements. These results will help to apply the theory in relation to applications mentioned in \cite{BS03,MR24}. Finally, we invite the interested reader to consult the full content of the present paper for a complete account of the formalization presented here.

\subsection{Acknowledgment} 
We thank the referee for a careful reading of the manuscript and for providing valuable comments that helped improve the exposition of the formalized results. We also thank the referee for pointing out several issues and typos in the mathematical presentation. This work has received funding from ISF grant 1577/23 (Einstein Institute of Mathematics, The Hebrew University of Jerusalem).

\section{General Setup}\label{sec:general-setup-graded}
A ring $R$ is said to be graded by a monoid $\iota$ if $R \cong \bigoplus_{i\in\iota} R_i$ where 
$R_i$ are subgroups of $R$ such that $R_i R_j \subseteq R_{i + j}$. 
When working with graded rings on paper, the two rings $R$ and $\bigoplus_{i\in\iota}R_i$ are often identified. 
However, during formalisation, in order to have an ergonomic framework for graded rings,  
we need a way to be able to talk about the two rings as different objects (typically for type theoretical reasons) while
maintaining the ability to switch between the two with ease. 
A full discussion of how graded ring is implemented in \Lean~can be found in \cite{WZ22}. 
The most general setup is as following:

\begin{lstlisting}[caption={Graded ring and modules}, mathescape=true, extendedchars=true]
variable {A M ιA ιM σA σM : Type*} 
variable [Ring A] [AddCommGroup M] [Module A M]
variable ($\mathcal{A}$ : ιA → σA) ($\mathcal{M}$ : ιM → σM)
variable [DecidableEq ιA] [AddMonoid ιA] [DecidableEq ιM]  [VAdd ιA ιM]
variable [SetLike σA A] [SetLike σM M]
variable [AddSubgroupClass σA A] [AddSubgroupClass σM M] 
variable [GradedRing $\mathcal{A}$] 
variable [DirectSum.Decomposition $\mathcal{M}$] [SetLike.GradedSMul $\mathcal{A}$ $\mathcal{M}$]
\end{lstlisting}
\lstinline[mathescape=true]|[SetLike σA A]| and \lstinline[mathescape=true]|[AddSubgroupClass σA A]| together assert that the terms of \lstinline|σA| are subgroups of $A$.
In particular, for each term $i$ of \lstinline[mathescape=true]|$\iota$A|, $\mathcal{A}_i$ is a subgroup of $A$. 
In the general setup, we do not use a concrete type like \lstinline[mathescape=true]|$\mcA$ : ιA → AddSubgroup A|; this is to avoid code duplication:
for example if we assume $A$ to be an $R$-algebra, by using \lstinline[mathescape=true]|[AddSubgroupClass σA A]|, we can specialize \lstinline|σA| to be $R$-submodules of $A$ and 
realize $A$ as a graded algebra without the needs to duplicate any general result.
\lstinline[mathescape=true]|[GradedRing $\mathcal{A}$]| is an abbreviation for \lstinline[mathescape=true]|[DirectSum.Decomposition $\mathcal{A}$]| and 
\lstinline[mathescape=true]|[SetLike.GradedMonoid $\mathcal{A}$]| where \lstinline[mathescape=true]|[SetLike.GradedMonoid $\mathcal{\McA}$]| is an abbreviation for 
\lstinline[mathescape=true]|[SetLike.GradedOne $\mcA$]| asserting that $1$ has grade zero and \lstinline[mathescape=true]|[SetLike.GradedMul $\mcA$]| asserting that $\mcA_i\mcA_j \subseteq \mcA_{i+j}$.
Similarly, \lstinline[mathescape=true]|[SetLike.GradedSMul $\mcA$ $\mcM$]| asserts that $\mcA_i \cdot \mcM_j \subseteq \mcM_{i+j}$ where $i + j$ is provided by \lstinline|[VAdd ιA ιM]|. 
The general setup here is versatile: by allowing \lstinline|ιA| and \lstinline|ιM| to be different types, we can have graded rings and modules that are not graded by the same monoid --- for example the ring is graded by $\NN$ and the module by $\ZZ$.

\subsection{Graded Ring Homomorphism}

Let $\iota$ be a monoid. Suppose $A$ and $B$ are two $\iota$-graded rings with grading $\mcA$ and $\mcB$ respectively.

\begin{lstlisting}[extendedchars=true]
variable {A B : Type*} [Ring A] [Ring B]
variable {ι : Type*} [DecidableEq ι] [AddMonoid ι]
variable {σA σB : Type*} [SetLike σA A] [SetLike σB B] 
variable [AddSubgroupClass σA A] [AddSubgroupClass σB B]
variable (𝒜 : ι → σA) [GradedRing 𝒜]
variable (ℬ : ι → σB) [GradedRing ℬ]
\end{lstlisting}
\begin{definition}\label{graded-ring-hom}
  A \emph{graded ring homomorphism} from $A$ to $B$ is a ring homomorphism $f : A \to B$ such that for all $i \in \iota$, 
  we have $f(\mcA_i) \subseteq \mcB_i$.
\begin{lstlisting}[caption={Graded ring homomorphism}, label={lst:graded-ring-hom}]
structure GradedRingHom extends RingHom A B where
  map_mem' : ∀ {i} {x : A}, 
    x ∈ 𝒜 i → toFun x ∈ ℬ i

scoped[Graded] infix:25 "→+*" => GradedRingHom
\end{lstlisting}
\begin{flushright}
  \cite[\lstinline|GradedRingHom| file location:{\footnotesize \textsf{Grading/GradedRingHom.lean}}]{projconstruction}
\end{flushright}
\end{definition}
In Listing \ref{lst:graded-ring-hom}, we keep the same notation for ring homomorphisms --- \lstinline[mathescape=true]|$\mcA$ →+* $\mcB$| is a graded ring homomorphism from $A$ to $B$.

To continue the philosophy that there should be easy ways to switch between the internally graded ring $A$ and the externally graded ring $\bigoplus_{i\in\iota} \mcA_i$, 
for a graded ring homomorphism $f : A \to B$, we define the corresponding ring homomorphism $\bigoplus_{i\in\iota}\mcA_i \to \bigoplus_{i\in\iota}\mcB_i$ by
\[
f_{\oplus} : \begin{tikzcd}[column sep=huge]
  \bigoplus_{i\in\iota} \mcA_i \arrow[r, "\recompose"] & A \arrow[r, "f"] & B \arrow[r, "\decompose"] & \bigoplus_{i\in\iota} \mcB_i.
\end{tikzcd}
\]
\begin{lstlisting}
def directSum (f : GradedRingHom 𝒜 ℬ) : (⨁ i, 𝒜 i) →+* (⨁ i, ℬ i) :=
  RingHom.comp (DirectSum.decomposeRingEquiv ℬ) <| 
    f.comp (DirectSum.decomposeRingEquiv 𝒜).symm
\end{lstlisting}

\begin{lemma}{}{}
  The kernel of a graded ring homomorphism $f : A \to B$ is a homogeneous ideal of $A$.
\begin{flushright}
\cite[\lstinline|GradedRingHom.ker| file location: {\footnotesize\textsf{Module/Graded/RingHom.lean}}]{jjaassoonn_dimensiontheory}
\end{flushright}
\end{lemma}
\begin{proof}{}
  Suppose $x = \sum_i x_i$ is mapped to $0$ by $f$ where each $x_i \in \mcA_i$. 
  We need to show $f(x_i)$ is zero as well. $f(x_i) = \left(f_{\oplus}\left( \bigoplus_i x_i \right)\right)_i = \decompose(f(x))_i = 0$.
\begin{lstlisting}[caption={Kernel of a graded ring homomorphism}, label={lst:ker-graded-ring-hom}]
def ker (f : GradedRingHom 𝒜 ℬ) : HomogeneousIdeal 𝒜 where
  __ := RingHom.ker f
  is_homogeneous' i x (hx : _ = 0) := show _ = 0 by
    simp [apply_decompose, ← decompose_apply, hx, RingHom.mem_ker] at hx ⊢
\end{lstlisting}
\end{proof}

\begin{remark}{}{}
Similarly, we define the notion of a graded ring isomorphism and a grade algebra homomorphism.
\begin{lstlisting}[extendedchars=true]
structure GradedRingEquiv extends RingEquiv A B where
  map_mem' : ∀ {i : ι} {x : A}, x ∈ 𝒜 i → toFun x ∈ ℬ i
  inv_mem' : ∀ {i : ι} {y : B}, y ∈ ℬ i → invFun y ∈ 𝒜 i

scoped[Graded] infix:25 "≃+*" => GradedRingEquiv

structure GradedAlgHom extends A →ₐ[R] B, GradedRingHom 𝒜 ℬ

scoped[Graded] notation:25 𝒜 " →ₐ[" R "] " ℬ => GradedAlgHom (R := R) 𝒜 ℬ
\end{lstlisting}
\end{remark}

\subsection{Lemmas about Homogeneous elements} 
The following lemma is used repetitively in the formalisation of multi-graded Proj construction.
\begin{lemma}\label{exists-homogeneous-of-dvd}
  Let $A$ be a commutative $\iota$-graded ring where $\iota$ is an abelian group. Suppose $a$ and $c$ are homogeneous elements of $A$ such that $a \mid c$.
  Then there exists a \emph{homogeneous} element $b$ such that $ab = c$.
\begin{flushright}
\cite[\lstinline|exists_homogeneous_of_dvd| file location: {\footnotesize\textsf{ForMathlib/SetLikeHomogeneous.lean}}]{projconstruction}
\end{flushright}
\end{lemma}
\begin{proof}{}
  Let $b = \dots + b_0 + b_1 + \dots $ be an arbitrary element of $A$ such that $ab = c$.
  Suppose $a$ has degree $i$ and $c$ has degree $j$, by looking at the $j$-th coordinate of the equaltion $ab = c$, 
  we see that $a b_{j - i} = c$ as well.
\end{proof}

\section{Relevance}\label{relevant-homogeneous-submonoid}\label{sec:relevant-homogeneous-submonoid}
In this section, we assume that $A$ is a commutative $\iota$-graded rings where $\iota$ is an abelian group.
We develop the notion a relevant homogeneous submonoid of $A$. 

\subsection{Homogeneous Submonoid}

A homogeneous submonoid of $A$ is a submonoid such that every element is homogeneous. Equivalently, it is a submonoid generated by homogeneous elements. Since this equivalent characterization is more convenient for our purposes, we adopt the following definition.

\begin{definition}{(Homogeneous Submonoid)}{}
  A \emph{homogeneous submonoid} of a graded ring $A$ is a submonoid $S$ of $A$ such that $S$ can be generated by homogeneous elements.
\begin{lstlisting}[caption={Homogeneous submonoid}]
structure HomogeneousSubmonoid extends Submonoid A where
  homogeneous_gen : ∃ (s : Set A),
    toSubmonoid = Submonoid.closure s ∧ ∀ x ∈ s, SetLike.IsHomogeneousElem 𝒜 x
\end{lstlisting}
\begin{flushright}
\cite[\lstinline|HomogeneousSubmonoid| file location: {\footnotesize\textsf{HomogeneousSubmonoid/Basic.lean}}]{projconstruction}
\end{flushright}
\end{definition}
By definition, if a set $s$ only contains homogeneous elements, then the submonoid $\left\langle  s \right\rangle$ generated by $s$ is a homogeneous submonoid.
\begin{lstlisting}
def closure (s : Set A) (hs : ∀ x ∈ s, SetLike.Homogeneous 𝒜 x) : HomogeneousSubmonoid 𝒜 where
  __ := Submonoid.closure s
  homogeneous_gen := ⟨s, rfl, hs⟩
\end{lstlisting}
The submonoid $\{1\}$ is a homogeneous submonoid of any graded ring $A$.

\begin{lemma}{}{}
  Let $S$ be a homogeneous submonoid of $A$ and $\Phi : A \to B$ be a graded ring homomorphism, then
  $\Phi_{\star}S := \left\langle \Phi(S) \right\rangle$ is a homogeneous submonoid of $B$.
\begin{flushright}
\cite[\lstinline|HomogeneousSubmonoid.map| file location: {\footnotesize\textsf{HomogeneousSubmonoid/Basic.lean}}]{projconstruction}  
\end{flushright}
\end{lemma}
\begin{proof}{}
  Suppose $S$ is generated by $s$ where $s$ is a set of homogeneous elements, $\left\langle \Phi(S) \right\rangle$ is generated by $\Phi(s)$.
  \begin{lstlisting}
def map (Φ : 𝒜 →+* ℬ) (S : HomogeneousSubmonoid 𝒜) : HomogeneousSubmonoid ℬ where
  toSubmonoid := S.toSubmonoid.map Φ
  homogeneous_gen := by
    obtain ⟨s, hs, h⟩ := S.homogeneous_gen
    refine ⟨Φ '' s, ?_, ?_⟩
    ...
  \end{lstlisting}
\end{proof}

\begin{lemma}{}{}
Let $S$ be a homogeneous submonoid, the set $\mybar{S}$ of homogeneous divisors of elements in $S$ is another homogeneous submonoid.
\begin{lstlisting}
def bar : HomogeneousSubmonoid 𝒜 where
  carrier := {x | SetLike.Homogeneous 𝒜 x ∧ ∃ y ∈ S, x ∣ y}
  mul_mem' := ...
  homogeneous_gen := ...
\end{lstlisting}
\begin{flushright}
\cite[\lstinline|HomogeneousSubmonoid.bar| file location: {\footnotesize\textsf{HomogeneousSubmonoid/Basic.lean}}]{projconstruction}  
\end{flushright}
\end{lemma}
\begin{proof}{}
  Obviously, $1\in\mybar S$ and $\mybar{S}$ contains only homogeneous elements.  
  Let $x$ and $y$ be two homogeneous elements in $\mybar{S}$, suppose $x$ divides $a_{x} \in S$ and $y$ divides $a_{y} \in S$.
  Then $xy$ is also homogeneous and divides $a_{x}a_{y} \in S$.
\end{proof}

\begin{construction}{}{}
    The collection of all homogeneous submonoids of a graded ring $A$ forms a commutative monoid.
    The multiplication is defined by the pointwise multiplication and the identity is the trivial homogeneous submonoid $\{1\}$.
    The product of two homogeneous submonoids $S$ and $T$ is still homogeneous because $ST$ is generated by the set $S \cup T$ 
    of homogeneous elements.

  We note that the following useful equalities: 
  \begin{itemize}
\item For any homogeneous submonoid $S$, $S \cdot S = S$.
\item If $s$ and $t$ are two sets of homogeneous elements,
    the closure of $s\cup t$ is equal to $\left\langle  s \right\rangle\left\langle  t \right\rangle$.
\item For any homogeneous submonoids $S$ and $T$ of $A$, if $\Phi : A \to B$ is a graded ring homomorphism, 
$\Phi_{\star}\left( ST \right) = \Phi_{\star}(S)\Phi_{\star}(T)$.
  \end{itemize}
\end{construction}

\begin{lstlisting}[extendedchars=true]
lemma mul_toSubmonoid (S T : HomogeneousSubmonoid 𝒜) : (S * T).toSubmonoid = S.toSubmonoid * T.toSubmonoid := rfl

lemma mul_self (S : HomogeneousSubmonoid 𝒜) : S * S = S 

instance : CommMonoid (HomogeneousSubmonoid 𝒜) where
  mul_assoc R S T:= toSubmonoid_injective _ <| mul_assoc _ _ _
  mul_comm S T :=  toSubmonoid_injective _ <| mul_comm _ _
  one_mul _ := toSubmonoid_injective _ <| one_mul _
  mul_one _ := toSubmonoid_injective _ <| mul_one _
  
lemma map_mul (Φ : 𝒜 →+* ℬ) (S T : HomogeneousSubmonoid 𝒜)  : (S * T).map Φ = S.map Φ * T.map Φ :=
\end{lstlisting}

\begin{flushright}
\cite[\lstinline|HomogeneousSubmonoid.instCommMonoid_1| file location: {\footnotesize\textsf{HomogeneousSubmonoid/Basic.lean}}]{projconstruction}
\end{flushright}

\subsection{Relevant Homogeneous Submonoid}

\begin{definition}{}{} \label{def:deg}
    Let $S$ be a homogeneous submonoid, we use $\deg(S)$ to denote the additive submonoid $\iota$ containing the degrees of elements in $S$,
    that is, $i \in \deg(S)$ if and only if there exists an $x\in S$ such that $x$ is homogeneous of degree $i$.

    We denote $\iota[S]$ to be the subgroup of $\iota$ generated by $\deg(S)$.
\begin{lstlisting}[extendedchars=true, caption={$\deg(S)$}]
def deg : AddSubmonoid ι where
  carrier := {i | ∃ x ∈ S, x ∈ 𝒜 i}
  add_mem' := ...
  zero_mem' := ...

def agrDeg : AddSubgroup ι := 
  AddSubgroup.closure S.deg

scoped notation:max ι"["S"]" => agrDeg (ι := ι) S
\end{lstlisting}
\begin{flushright}
\cite[\lstinline|deg| and \lstinline|agrDeg| file location: {\footnotesize\textsf{HomogeneousSubmonoid/Basic.lean}}]{projconstruction}
\end{flushright}
\end{definition}

\begin{remark}{}{}
  If $0$ is in $S$, $\deg S$ is equal to $\iota$.
  This is not problematic, because if $S$ contains $0$, in the context of localizations, $A_{S}$ is the trivial ring.
\end{remark}

\begin{remark}{}{}
As a set, $\iota[S]$ is the set of elements of the form $i - j$ where $i$ and $j$ are in $\deg(S)$.
\end{remark}

\begin{definition}{(Relevance)}{} \label{def:relev}
  \begin{itemize}
    \item A homogeneous submonoid $S$  is \emph{relevant} if for all $i \in \iota$, there exists a positive natural number $n$ such that 
  $n\cdot i$ is in $\iota\left[ \mybar S \right]$.
    \item A set of homogeneous elements $\left\{ a_i | i \in I \right\}$ is relevant if the homogeneous submonoid generated by $\left\{ a_i | i \in I \right\}$ is relevant.
    \item A homogeneous element $a$ is relevant if the set $\{a \}$ is relevant. 
    \item The homogeneous ideal of $A$ generated by the set of relevant homogeneous elements of $A$ is denoted as $A_{\dagger}$.
  \end{itemize}
\begin{lstlisting}[caption={Relevant homogeneous submonoid}, extendedchars=true]
def IsRelevant : Prop := ∀ (i : ι), ∃ (n : ℕ), 0 < n ∧ n • i ∈ ι[S.bar]

abbrev SetIsRelevant (s : Set A) (hs : ∀ i ∈ s, SetLike.Homogeneous 𝒜 i) : Prop :=
  HomogeneousSubmonoid.closure s hs |>.IsRelevant

abbrev ElemIsRelevant (a : A) (ha : SetLike.Homogeneous 𝒜 a) : Prop :=
  HomogeneousSubmonoid.closure {a} (by simpa) |>.IsRelevant

def dagger : HomogeneousIdeal 𝒜 where
  __ := Ideal.span { x | ∃ (h : SetLike.Homogeneous 𝒜 x), ElemIsRelevant x h }
  is_homogeneous' := Ideal.homogeneous_span _ _ (by rintro x ⟨h, _⟩; exact h)

scoped postfix:max "†" => dagger
\end{lstlisting}
\begin{flushright}
\cite[\lstinline|IsRelevant|, \lstinline|SetIsRelevant|, \lstinline|ElemIsRelevant|, and \lstinline|dagger| file location: {\footnotesize\textsf{HomogeneousSubmonoid/\{Relevant, Dagger\}.lean}}]{projconstruction}
\end{flushright}
\end{definition}

\begin{remark}{}{}
  A more succinct way to say that a homogeneous submonoid $S$ is relevant is that the quotient $\quotient{\iota}{\iota\left[ \mybar S \right]}$ is 
  a torsion abelian group.
  Hence, when $\iota$ is finitely generated, $S$ is relevant if and only if $\quotient{\iota}{\iota\left[ \mybar S \right]}$ is finite
  if and only if $\iota\left[ \mybar S \right]$ is subgroup of finite index.
\end{remark}

We begin with several lemmas about relevant homogeneous submonoids.
\begin{lemma}{}{relevant-mul}
  If $S$ and $T$ are two relevant homogeneous submonoids, then $ST$ is also relevant.
\begin{flushright}
\cite[\lstinline|IsRelevant.mul| file location: {\footnotesize\textsf{HomogeneousSubmonoid/Relevant.lean}}]{projconstruction}
\end{flushright}
\end{lemma} 
\begin{proof}{}
  Let $i \in \iota$, since $S$ is relevant, there exists a positive natural number $m$ such that $m\cdot i$ is in $\iota\left[\mybar S\right]$.
  Since $T$ is relevant, there exists a positive natural number $n$ such that $n\cdot i$ is in $\iota\left[\mybar T\right]$.
  One can show that $(m + n)\cdot i$ is in $\iota\left[ \mybar{ST} \right]$:
  since $m\cdot i \in \iota\left[\mybar S\right]$, we can find $a$ and $b$ in $\deg\left( \mybar S \right)$ such that $m\cdot i = a - b$.
  Similarly, we can find $c$ and $d$ in $\deg\left( \mybar T \right)$ such that $n\cdot i = c - d$.
  Then $(a + c) - (b + d) = (m + n)\cdot i$ is in $\iota\left[ \mybar{ST} \right]$.
\end{proof}

\begin{lemma}{}{} 
  Let $f : A\to B$ be a graded ring homomorphism and $S$ be a relevant homogeneous submonoid of $A$, $f_\star S = \left\langle f(S) \right\rangle$ 
  is a relevant homogeneous submonoid of $B$.
\begin{flushright}
\cite[\lstinline|IsRelevant.map| file location: {\footnotesize\textsf{HomogeneousSubmonoid/Relevant.lean}}]{projconstruction}
\end{flushright}
\end{lemma}
\begin{proof}{}
  It is sufficient show that $\iota\left[\mybar S\right] \le \iota\left[ \mybar{f_\star S} \right]$; or equivalently 
  $\deg\left( \mybar S \right) \le \deg\left( \mybar{f_\star S} \right)$.
  Let $i \in \deg\left( \mybar S \right)$. There exists an $x \in \mybar S$ such that $x$ is homogeneous of degree $i$.
  Hence, there exists an $y \in S$ such that $x$ divides $y$.
  Since $f$ is graded, $f(x)$ also has degree $i$ and $f(y) \in f_\star S$ is homogeneous and $f(x)$ divides $f(y)$. 
  Therefore, $f(x)$ is in $\mybar{f_\star S}$ and $i$ is in $\deg\left( \mybar{f_\star S} \right)$.
\end{proof}

\begin{lemma}{}{}
  If $T \le S$ are two relevant homogeneous submonoids where $T$ is relevant, then $S$ is also relevant.
\begin{flushright}
\cite[\lstinline|IsRelevant.ofLE| file location: {\footnotesize\textsf{HomogeneousSubmonoid/Relevant.lean}}]{projconstruction}
\end{flushright}
\end{lemma}
\begin{proof}
  Since $T$ is relevant, $\quotient{\iota}{\iota\left[ \mybar{T} \right]}$ is a torsion abelian group.
  Since $T \le S$, we have a surjection $\quotient{\iota}{\iota\left[ \mybar{T} \right]} \to \quotient{\iota}{\iota\left[ \mybar{S} \right]}$.
  Hence, $\quotient{\iota}{\iota\left[ \mybar{S} \right]}$ is also a torsion abelian group.
\end{proof}

\begin{theorem}\label{elemIsRelevant-iff}
  Suppose $\iota$ is a finitely generated abelian group. 
  Then $a \in A$ is a relevant homogeneous element if and only if 
  there exists elements $x_1, \dots, x_n$ in $A$ where each $x_i$ is homogeneous of degree $d_i$ such that
  the subgroup generated by $\left\{ d_i \right\}$ is of finite index in $\iota$ and there exists a natural number $k$ such that $a^k = \prod_{i=1}^{n}x_i$.
\begin{lstlisting}[extendedchars=true]
lemma elemIsRelevant_iff [AddGroup.FG ι] (a : A) (ha : SetLike.Homogeneous 𝒜 a) :
    ElemIsRelevant a ha ↔
    ∃ (n : ℕ) (x : Fin n → A) (d : Fin n → ι)
      (_ : ∀ (i : Fin n), x i ∈ 𝒜 (d i)),
      (AddSubgroup.closure (Set.range d)).FiniteIndex ∧
      (∃ (k : ℕ), ∏ i : Fin n, x i = a ^ k) := ...
\end{lstlisting}
\begin{flushright}
\cite[\lstinline|elemIsRelevant_iff| file location: {\footnotesize\textsf{HomogeneousSubmonoid/Relevant.lean}}]{projconstruction}
\end{flushright}
\end{theorem}
\begin{proof}{}
  For the forward implication, suppose $a$ is a relevant homogeneous element, we see that $\iota\left[ \mybar{\left\langle a \right\rangle} \right]$ 
  is a subgroup of finite index in $\iota$. Therefore, $\iota\left[ \mybar{\left\langle a \right\rangle} \right]$ is finitely generated.
  Let $\left\{ i_1,\dots,i_N \right\} \subseteq \iota$ be a generating set of $\iota\left[ \mybar{\left\langle a \right\rangle} \right]$.
  For each $i_k \in \left\{ i_1, \dots, i_n \right\}$, we see that there exists an $x_k \in A$ homogeneous of degree $i_k$ such that $x_l$ divides $a^n_k$ for some $n_j$.
  We can take $K = \sum_{k=1}^N n_k$, and we see that $\prod_{k=1}^{N}x_k$ divides $a^K$.
  Hence, by Lemma \ref{exists-homogeneous-of-dvd}, there exists a homogeneous element $b$ of degree $j$ such that 
  $\left( \prod_{k=1}^{N} x_k\right) b = a^K$.
  All there remains is to show that the subgroup generated by $\left\{ i_1, \dots, i_N, j \right\}$ is of finite index in $\iota$.
  Since $\iota\left[ \mybar{\left\langle a \right\rangle} \right] = \left\langle i_1, \dots, i_N \right\rangle$ is a subgroup of finite index in $\iota$,
  and $\left\langle i_1, \dots, i_N, j \right\rangle$ is at least as large as, if not larger than, $\left\langle i_1, \dots, i_N \right\rangle$, 
  it is also a subgroup of finite index in $\iota$.

  Conversely, suppose $a^k = \prod_{i=1}^n x_i$ such that $x_i$ is homogeneous of degree $d_i$ and the subgroup generated by $\left\{ d_i \right\}$ is of finite index in $\iota$.
  It is sufficient to show that $\left\langle d_1,\dots, d_n \right\rangle$ is smaller than or equal to $\iota\left[ \mybar{\left\langle a \right\rangle} \right]$.
  This is true because for any $i$, $d_i$ is in $\deg\left( \mybar{\left\langle a \right\rangle} \right)$: 
    due to the factorization $\prod_i x_i = a^k$, $x_i$ has degree $d_i$ and is a homogeneous divisor of $a^k \in \left\langle a \right\rangle$.
\end{proof}

\begin{corollary}{}{}
  If $a$ and $b$ are two relevant homogeneous elements, $ab$ is also a relevant homogeneous element.
\begin{flushright}
\cite[\lstinline|ElemIsRelevant.mul| file location: {\footnotesize\textsf{HomogeneousSubmonoid/Relevant.lean}}]{projconstruction}
\end{flushright}
\end{corollary}

\begin{corollary}\label{graded-ring-hom-image-relevant}
  If $f : A \to B$ is a graded ring homomorphism and $a$ is a relevant homogeneous element in $A$, $f(a)$ is a relevant homogeneous element in $B$.
\begin{flushright}
\cite[\lstinline|GradedRingHom.map_relevant| file location: {\footnotesize\textsf{HomogeneousSubmonoid/Dagger.lean}}]{projconstruction}
\end{flushright}
\end{corollary}

\begin{corollary}\label{dagger-map-le}
  Let $f : A \to B$ be a graded ring homomorphism, $f_{\star}\left( A_{\dagger} \right)$ is smaller than or equal to $B_{\dagger}$.
\begin{flushright}
\cite[\lstinline|GradedRingHom.map_dagger_le| file location: {\footnotesize\textsf{HomogeneousSubmonoid/Dagger.lean}}]{projconstruction}
\end{flushright}
\end{corollary}

\begin{corollary}{}{}
  Let $f : A \to B$ be a surjective graded ring homomorphism, the radical ideals of $f_{\star}\left( A_{\dagger} \right)$ and $B_{\dagger}$ are equal :
  \[
  \rad{f_{\star}\left( A_{\dagger} \right)} = \rad{B_{\dagger}}.
  \]
\begin{lstlisting}
lemma radical_dagger_eq_of_surjective (surj : Function.Surjective f) :
    ((𝒜 †).toIdeal.map f).radical = (ℬ †).toIdeal.radical := by
\end{lstlisting}
\begin{flushright}
\cite[\lstinline|GradedRingHom.radical_dagger_eq_of_surjective| file location: {\footnotesize\textsf{HomogeneousSubmonoid/Dagger.lean}}]{projconstruction}
\end{flushright}
\end{corollary}
\begin{proof}{}
  By Corollary \ref{dagger-map-le}, we have $f_{\star}\left( A_{\dagger} \right) \subseteq B_{\dagger}$, hence we only need to show 
  $\rad{B_{\dagger}} \le \rad{f_{\star}\left( A_{\dagger} \right)}$. 
  Equivalently, we need to show that for any relevant homogeneous element $b \in B$, $b$ is in $\rad{f_{\star}\left( A_{\dagger} \right)}$.
  Since $b$ is relevant, by Theorem \ref{elemIsRelevant-iff}, there exists a positive natural number $k$ such that $b^k = \prod_{i=1}^{n} x_i$ where each $x_i$ is homogeneous of degree $d_i$ and 
  $\left\langle  d_i \right\rangle$ is a subgroup of finite index.
  Since $f$ is surjective, we can write $b$ as $f(a)$ for some $a \in A$ and $f(a_i) = x_i$ for each $i = 1,\dots, n$. 
  In order to show that $b = f(a)$ is in $\rad{f_{\star}\left( A_{\dagger} \right)}$, we show that $f(a)^k = f\left(\prod_i a_i\right)$ is in $f_{\star}\left( A_{\dagger} \right)$.
  This is true, because $\prod_i a_i$ is relevant by Theorem \ref{elemIsRelevant-iff}.
\end{proof}

\section{Homogeneous Localization}\label{sec:homogeneous-localization}
In this section, let $R$ be a commutative ring and $A$ a commutative graded $R$-algebra. Let $x$ be a submonoid of $A$. 
\begin{lstlisting}[extendedchars=true]
variable {ι R A : Type*}
variable [CommRing R] [CommRing A] [Algebra R A]
variable (𝒜 : ι → Submodule R A)
variable (x : Submonoid A)
\end{lstlisting}

In general, unless $x$ contains only homogeneous elements, the localized ring $A_x$ is not graded\footnote{When $x$ is not the complement of 
a prime ideal, the better notation for localized ring is $x^{-1}A$. But we use this notation to keep align with homogeneous localization defined below}. 
However, we can still investigate set $A_{(x)}$ of elements of degree zero in the localized ring. 

\begin{construction}
      Consider the type $\mathcolorbox{yellow!10}{P}$ of triples of $(i, a, b)$ where $a$ and $b$ are 
  homogeneous elements of the degree $i$ and $b$ is in $x$, we have a function $\mathcolorbox{blue!10}{\mathsf{frac} : P \to A_{x}}$ defined by
  $(i, a, b) \mapsto =\frac a b$.
  We define the homogeneous localization $A_{(x)}$ to be the quotient $\quotient{P}{\sim}$ where \sethlcolor{red!10}\hl{two pairs $p$ and $q$
  are equivalent if and only if $\mathsf{frac}(p) = \mathsf{frac}(q)$}. 
  The function $\mathsf{frac}$ descends to an embedding $\mathcolorbox{green!10}{\iota: \hloca A{x} \hookrightarrow A_{x}}$.
\begin{lstlisting}[linebackgroundcolor={%
  \ifnum\value{lstnumber}=1\color{yellow!10}\fi
  \ifnum\value{lstnumber}=2\color{yellow!10}\fi
  \ifnum\value{lstnumber}=3\color{yellow!10}\fi
  \ifnum\value{lstnumber}=4\color{yellow!10}\fi
  \ifnum\value{lstnumber}=6\color{blue!10}\fi
  \ifnum\value{lstnumber}=7\color{blue!10}\fi
  \ifnum\value{lstnumber}=8\color{blue!10}\fi
  \ifnum\value{lstnumber}=13\color{green!10}\fi
  \ifnum\value{lstnumber}=14\color{green!10}\fi
  \ifnum\value{lstnumber}=15\color{green!10}\fi},
  mathescape=true]
structure NumDenSameDeg where
  deg : ι
  (num den : 𝒜 deg)
  den_mem : (den : A) ∈ x

def embedding (p : NumDenSameDeg 𝒜 x) : 
    Localization x :=
  Localization.mk p.num ⟨p.den, p.den_mem⟩

def HomogeneousLocalization : Type _ :=
  Quotient <| (*@\sethlcolor{red!10}\hl{Setoid.ker <| embedding $\mcA$ x}@*)

def val (y : HomogeneousLocalization 𝒜 x) : 
    Localization x :=
  Quotient.liftOn' y (embedding 𝒜 x) ...
\end{lstlisting}
\begin{flushright}
\cite[\href{https://leanprover-community.github.io/mathlib4_docs/Mathlib/RingTheory/GradedAlgebra/HomogeneousLocalization.html\#HomogeneousLocalization}{\lstinline|HomogeneousLocalization|}]{mathlib4docs}
\end{flushright}
\end{construction}

We choose not to define $\hloc A x$ as a subring of the localized ring $A_x$. 
In the current approach, for any fraction $f \in \hloc A x$, we can use the \lstinline|induction| tactic to obtain a pair $(i, a, b)$ such that $f = [(i, a ,b)]$. 
Hence, the numerator, the denominator, and the degree of them are easily available and more organized.

\begin{theorem}\label{homogeneous-localization-local-ring}
  Let $\mfp$ be a prime ideal of $A$, the homogeneous localization $\hloc A {\mfp}$ is a local ring.
\begin{flushright}
\cite[\href{https://leanprover-community.github.io/mathlib4_docs/Mathlib/RingTheory/GradedAlgebra/HomogeneousLocalization.html\#HomogeneousLocalization.isLocalRing}{\lstinline|HomogeneousLocalization.isLocalRing|}]{mathlib4docs}
\end{flushright}
\end{theorem}
\begin{proof}{}
Since $A_{\mfp}$ is a local ring, it is sufficient to prove that for any fraction $f \in \hloc A{\mfp}$, $f$ is a unit in $A_{\mfp}$ if and only if $f$ is a unit in $\hloc A{\mfp}$ as well.
This is because, we write $f=[(i, a, b)]$, the inverse of $f$ is $[i, b, a]$.
\end{proof}

Homogeneous localization of a graded module can be defined similarly.

\begin{lemma}\label{hloc-map}
  Let $\phi : A \to B$ be a graded ring homorphism, suppose $P$ is a submonoid of $A$ and $Q$ is a submonoid of $B$ are submonoids such that $\phi^{-1}(Q) \le P$.
  Then $\phi$ induces a ring homomorphism $\hloc A P \to \hloc B Q$.
\begin{flushright}
\cite[HomogeneousLocalization.map]{mathlib4docs}
\end{flushright}
\end{lemma}
\begin{proof}{}
  $\phi$ induces a map $A_{P} \to A_{Q}$; since $\phi$ preserves degrees, we can restrict the map to the degree zero part of the localized rings.
\end{proof}
\begin{remark}{}
  Since the homogeneous localization is not implemented as a subring of the localized ring, we can not restrict $\phi$ to the degree zero part; we have to construct the map by hand:
\begin{lstlisting}[extendedchars=true]
def HomogeneousLocalization.map (g : A →+* B)
    (comap_le : P ≤ Q.comap g) (hg : ∀ i, ∀ a ∈ 𝒜 i, g a ∈ ℬ i) :
    HomogeneousLocalization 𝒜 P →+* HomogeneousLocalization ℬ Q where
  toFun := Quotient.map' (fun x ↦ ⟨x.1, ⟨_, hg _ _ x.2.2⟩, ⟨_, hg _ _ x.3.2⟩, comap_le x.4⟩) ...
  map_add' := ...
  map_mul' := ...
  map_zero' := ...
  map_one' := ...
\end{lstlisting}
Luckily, with \lstinline|NumDenSameDeg| structure, the numerator and the denominator are easily accessible.
\end{remark}

\section{Localization of Graded Ring and Modules}\label{graded-loc}\label{sec:graded-loc}
In section \ref{sec:homogeneous-localization}, we see that not every localized ring is graded. 
If we restrict our attention to homogeneous submonoids, we can define the grading of the localized rings and modules.
In this section, we assume the indexing set for the grading of $A$ is an additive group. 
Let $S$ be a homogeneous submonoid of $A$, we construct the quotient grading for the localized ring $A_S$. 
We will adopt a similar approach to the construction of homogeneous localization in section \ref{sec:homogeneous-localization}.

\begin{construction}\label{grading-localization}
Let $\mathcolorbox{yellow!10}{P_i}$ denote the set of qudraples $(a, b, d_a, d_b)$ where $d_a$ and $d_b$ are in $\iota$ such that $d_a - d_b = i$, 
$a \in A$ is homogeneous of degree $d_a$ and $b \in S$ is homogeneous of degree $d_b$.
For each $i \in \iota$, we give $P_i$ a group structure by: 
\begin{itemize}
  \item The zero element is $(0, 1, 0, 0)$.
  \item The result of adding $(a, b, d_a, d_b)$ and $(a', b', d_{a'}, d_{b'})$ is defined as
$(ab' + a'b, bb', d_a + d_{a'}, d_b + d_{b'})$.
  \item The negation of $(a, b, d_a, d_b)$ is defined as $(-a, b, d_a, d_b)$.
\end{itemize}
Hence, we have \sethlcolor{red!10}\hl{a group homomorphism $\mathsf{frac}: P_i \to A_{S}$ defined by $(a, b, d_a, d_{b}) \mapsto \frac{a}{b}$}.
Consider the quotient formed by $\quotient{P_{i}}{\sim}$ where \sethlcolor{green!10}\hl{two quadruples $p$ and $q$ are equivalent if and only if $\mathsf{frac}(p) = \mathsf{frac}(q)$}.
We take \sethlcolor{olive!10}\hl{the $i$-th grading $A_{S,i}$ of $A_S$ to be the range of the group homomorphism} \sethlcolor{purple!15}\hl{$\mathsf{frac} : \quotient{P_{i}}{\sim} \to A_{S}$}.
\begin{lstlisting}[linebackgroundcolor={%
  \ifnum\value{lstnumber}=10\color{red!10}\fi
  \ifnum\value{lstnumber}=11\color{red!10}\fi
  \ifnum\value{lstnumber}=12\color{red!10}\fi
  \ifnum\value{lstnumber}=13\color{red!10}\fi
  \ifnum\value{lstnumber}=14\color{red!10}\fi
  \ifnum\value{lstnumber}=15\color{red!10}\fi
  \ifnum\value{lstnumber}=17\color{green!10}\fi
  \ifnum\value{lstnumber}=18\color{green!10}\fi
  \ifnum\value{lstnumber}=19\color{green!10}\fi
  \ifnum\value{lstnumber}=21\color{purple!15}\fi
  \ifnum\value{lstnumber}=22\color{purple!15}\fi
  \ifnum\value{lstnumber}=23\color{purple!15}\fi
  \ifnum\value{lstnumber}=24\color{purple!15}\fi
  \ifnum\value{lstnumber}=26\color{olive!10}\fi
  \ifnum\value{lstnumber}=27\color{olive!10}\fi
  \ifnum\value{lstnumber}=28\color{olive!10}\fi
}, caption={Grading of localized rings}]
structure (*@\sethlcolor{yellow!10}\hl{PreLocalizationGrading}@*) (i : ι) where
  num : A
  den : S.toSubmonoid
  degNum : ι
  degDen : ι
  num_mem : num ∈ 𝒜 degNum
  den_mem : (den : A) ∈ 𝒜 degDen
  deg_frac_eq : degNum - degDen = i

def val (i : ι) : 
    (S.PreLocalizationGrading i) →+ 
    Localization S.toSubmonoid where
  toFun x := Localization.mk x.num x.den
  map_zero' := Localization.mk_zero 1
  map_add' := by simp [Localization.add_mk]

def addCon (i : ι) : 
    AddCon (S.PreLocalizationGrading i) := 
  AddCon.ker (val S i)

def emb (i : ι) : 
    (addCon S i).Quotient →+ 
    Localization S.toSubmonoid :=
  AddCon.lift _ (val ..) le_rfl

def LocalizationGrading (i : ι) : 
    AddSubgroup (Localization S.toSubmonoid) := 
  (PreLocalizationGrading.emb S i).range
\end{lstlisting} 
To construct the decomposition ring homomorphism $A_{S} \to \bigoplus_i A_{S, i}$, we first construct a ring homomorphism $A \to \bigoplus_{i} A_{S, i}$ and check that 
every element in $S$ is sent to an invertible element. 
The ring homomorphism is defined as the following:
\[
\begin{tikzcd}[column sep=large]
A \arrow[r, "\mathcolorbox{gray!10}{\decompose}"] & \bigoplus_{i\in\iota} A_i \arrow[r, "\mathcolorbox{brown!10}{{\bigoplus x \mapsto \frac x 1}}"] & \bigoplus_{i\in\iota} A_{S,i}.
\end{tikzcd}
\]

\begin{lstlisting}[extendedchars=true, linebackgroundcolor={%
  \ifnum\value{lstnumber}=5\color{brown!10}\fi
  \ifnum\value{lstnumber}=6\color{brown!10}\fi
  \ifnum\value{lstnumber}=7\color{brown!10}\fi
  \ifnum\value{lstnumber}=8\color{brown!10}\fi
  \ifnum\value{lstnumber}=9\color{brown!10}\fi
  \ifnum\value{lstnumber}=10\color{brown!10}\fi
  \ifnum\value{lstnumber}=11\color{brown!10}\fi
  \ifnum\value{lstnumber}=12\color{brown!10}\fi
  \ifnum\value{lstnumber}=13\color{brown!10}\fi
  },
  linebackgroundwidth=370pt,
  linebackgroundsep=-24pt]
noncomputable def decomposition :
    Localization S.toSubmonoid →+* ⨁ i : ι, S.LocalizationGrading i :=
  IsLocalization.lift (M := S.toSubmonoid) (S := Localization S.toSubmonoid)
    (g := RingHom.comp
      (DirectSum.toSemiring (fun i ↦
        (DirectSum.of (fun i ↦ S.LocalizationGrading i) i : 
          S.LocalizationGrading i →+ ⨁ i, S.LocalizationGrading i).comp
        (⟨⟨fun x ↦ ⟨Localization.mk x.1 1,
          ⟨AddCon.mk' _ ⟨x.1, 1, i, 0, x.2, SetLike.GradedOne.one_mem, by simp⟩, rfl⟩⟩, 
            ...⟩, ...⟩ : 
            𝒜 i →+ S.LocalizationGrading i))
            ... 
            ...) 
      (*@\sethlcolor{gray!10}\hl{(DirectSum.decomposeRingEquiv $\mcA$).toRingHom}@*))
  ... --- check that elements in (*@$S$@*) are sent to invertible elements
\end{lstlisting}
To summarize, an element $x \in A_{S}$ has degree $i$ if and only if there exists $m,n\in\iota$ such that $m - n = \iota$ and 
$x = \frac{a}{b}$ for some $a \in A$ with degree $m$ and $b \in S$ with degree $n$.
\begin{lstlisting}[extendedchars=true]
lemma mem_localizationGrading_iff (x : Localization S.toSubmonoid) (i : ι) :
    x ∈ S.LocalizationGrading i ↔
    ∃ (m n : ι) (_ : m - n = i) (a : 𝒜 m) (b : 𝒜 n) (hb : b.1 ∈ S.toSubmonoid),
    x = Localization.mk a.1 ⟨b, hb⟩ := by
  constructor
  · rintro ⟨x, rfl⟩
    obtain ⟨x, rfl⟩ := AddCon.mk'_surjective x
    exact ⟨x.degNum, x.degDen, x.deg_frac_eq, ⟨x.num, x.num_mem⟩, ⟨x.den, x.den_mem⟩,
      x.den.2, rfl⟩
  · rintro ⟨m, n, rfl, ⟨a, ha⟩, ⟨b, hb⟩, hb', rfl⟩
    exact ⟨AddCon.mk' _ ⟨a, ⟨b, hb'⟩, m, n, ha, hb, rfl⟩, rfl⟩
\end{lstlisting}
\begin{flushright}
\cite[file location: {\footnotesize\textsf{Grading/Localization.lean}}]{projconstruction}
\end{flushright}
\end{construction}

If $M$ is a graded $A$-module also with $\iota$ as the indexing set for the grading,
we can construct the localized module $M_{S}$ as a graded $A_{S}$-module in the same way as Construction \ref{grading-localization}.



As graded rings, we can always enlarge the submonoid $S$ to $\mybar S$ and have the rings $A_S$ isomorphic to $A_{\mybar S}$:
\begin{theorem}\label{loc-iso-bar}
  The localized rings $A_{S}$ and $A_{\mybar S}$ are isomorphic as graded rings.
\begin{flushright}
\cite[\lstinline|localizationEquivLocalizationBar| file location: {\footnotesize\textsf{HomogeneousSubmonoid/IsoBar.lean}}]{projconstruction}
\end{flushright}
\end{theorem}
\begin{proof}{}
  Since $S \le \mybar{S}$, we have a ring homomorphism $\leanmath{localizationToLocalizationBar}: A_{S} \to A_{\mybar S}$.
  Let $s \in \mybar S$, there exists an $y \in S$ such that $s$ divides $y$.
  
  In order to descend the map $A \to A_S$ to a map $A_{\mybar S} \to A_S$, we need to check each element $s \in \mybar{S}$ is invertible in $A_{S}$.
  By~\cref{exists-homogeneous-of-dvd}, there exists a homogeneous element $z$ such that $sz = y$.
  Therefore, the fraction $\frac s 1 \in A_{S}$ has an inverse $\frac{z}{y} \in A_{S}$. 
  Hence, we have a well-defined ring homomorphism $\leanmath{localizationBarToLocalization}: A_{\mybar S} \to A_{S}$. 
  The two ring homomorphisms are inverses of each other.

  We show that the ring isomorphism $A_{S} \to A_{\mybar S}$ is a graded:
  suppose $x \in A_{S}$ is homogeneous of degree $i$, by~\cref{grading-localization}, $x$ can be written as $\frac{a}{b}$
  with $a$ homogeneous of degree $m$ and $b \in S$ homogeneous of degree $n$ such that $m - n = i$.
  Thus image of $x$ in $A_{S}$ is still $\frac a b$ where $b$ is seen as an element of $\mybar S$; hence $x$ still has degree $i$ in $A_{\mybar S}$.
  Similarly, if $x \in A_{\mybar S}$ is homogeneous of degree $i$, then it can be written as $\frac{a}{b}$ where $a$ is homogeneous of degree $m$ and 
  $b \in \mybar S$ is homogeneous of degree $n$ such that $m - n = i$. 
  Suppose $b$ divides $y \in S$, then there exists a homogeneous element $z$ such that $bz = y$.
  Therefore, the image of $x$ in $A_{S}$ is $\frac{a z}{b z}$, which is still homogeneous of degree $i$.

\begin{lstlisting}[extendedchars=true, caption={$A_S \cong A_{\mybar S}$ as graded rings}]
def localizationToLocalizationBar : Localization S.toSubmonoid →+* Localization S.bar.toSubmonoid :=
IsLocalization.lift
  (M := S.toSubmonoid) (R := A) (S := Localization S.toSubmonoid) (g := algebraMap _ _) ...

def localizationBarToLocalization : Localization S.bar.toSubmonoid →+* Localization S.toSubmonoid :=
IsLocalization.lift
  (M := S.bar.toSubmonoid) (R := A) (S := Localization S.bar.toSubmonoid) (g := algebraMap _ _) ...

def localizationEquivLocalizationBar : S.LocalizationGrading ≃+* S.bar.LocalizationGrading where
    __ := S.localizationToLocalizationBar
  invFun := S.localizationBarToLocalization
  left_inv := ...
  right_inv := ...
  map_mem' := ...
  inv_mem' := ...
\end{lstlisting}
\end{proof}

\section{Tensor Product of Graded Rings}\label{graded-tensor}\label{sec:graded-tensor}
In this section, we assume $A$ and $B$ are graded commutative $R$-algebras where $A$ is graded by $\iota_{A}$ and $B$ is graded by $\iota_B$. 
We will realize $A \ox_R B$ as an $\iota_A \times \iota_B$-graded $R$-algebra.
\begin{lstlisting}[extendedchars=true]
variable {ιA ιB R A B : Type*}
variable [DecidableEq ιA] [AddCommGroup ιA]
variable [DecidableEq ιB] [AddCommGroup ιB]
variable [CommRing R] [CommRing A] [Algebra R A] [CommRing B] [Algebra R B]
variable (𝒜 : ιA → Submodule R A) (ℬ : ιB → Submodule R B)
variable [GradedAlgebra 𝒜] [GradedAlgebra ℬ]
\end{lstlisting}

\begin{construction}{(Graded Tensor Product)}
  For any $i \in \iota_A$ and $j \in \iota_B$, the $(i,j)$-th graded piece of the tensor product $A \ox_R B$ is the 
  $A_i \ox_R B_j$. Since type theoretically, we need $A_i \ox_R B_j$ as a submodule of $A \ox_R B$, we use the range 
  of the linear map $A_i \ox_R B_j \to A \ox_R B$.
  \begin{lstlisting}[caption={Graded Tensor Product}]
noncomputable def gradingByProduct : ιA × ιB → Submodule R (A ⊗[R] B) := fun p ↦
  LinearMap.range (TensorProduct.map (𝒜 p.1).subtype (ℬ p.2).subtype)

scoped infix:min "⊗" => gradingByProduct
  \end{lstlisting}
We can check that $1$ has degree $(0, 0)$ and that the product of an element of degree $(i, j)$ and an element of degree $(i', j')$ has degree $(i + i', j + j')$.
To show that \lstinline|gradingByProduct| is indeed a grading, we construct a linear map 
  \[
  \begin{tikzcd}
  A \ox_R B \arrow[r, "\mathcolorbox{yellow!10}{\cong}"] & {\bigoplus_i\!A_i \; \ox_R \; \bigoplus_j\!B_j} \arrow[r, "\bigoplus \ox"] & \bigoplus_{(i, j)} A_i \ox_R B_j
  \end{tikzcd},
  \]
  where the second map is defined as 
  \[
    \mathcolorbox{red!10}{
    (0,\dots,0,\underset{\mathclap{\substack{\uparrow \\ \text{$i$-th index}}}}{a},0,\dots,0) \ox 
    (0,\dots,0,\underset{\mathclap{\substack{\uparrow \\ \text{$j$-th index}}}}{b},0,\dots,0) \mapsto 
    (0,\dots,0,\underset{\mathclap{\substack{\uparrow \\ \text{$(i, j)$-th index}}}}{a \ox b},0,\dots,0)}.
  \]
\begin{lstlisting}[linebackgroundcolor={%
  \ifnum\value{lstnumber}=3\color{red!10}\fi
  \ifnum\value{lstnumber}=4\color{red!10}\fi
  \ifnum\value{lstnumber}=5\color{red!10}\fi
  \ifnum\value{lstnumber}=6\color{red!10}\fi
  \ifnum\value{lstnumber}=7\color{red!10}\fi
  \ifnum\value{lstnumber}=8\color{red!10}\fi
  \ifnum\value{lstnumber}=9\color{red!10}\fi
  \ifnum\value{lstnumber}=11\color{yellow!10}\fi
  \ifnum\value{lstnumber}=12\color{yellow!10}\fi
  \ifnum\value{lstnumber}=13\color{yellow!10}\fi},
  linebackgroundsep=-8pt,
  linebackgroundwidth=400pt, extendedchars=true]
noncomputable def decompositionByProduct : (A ⊗[R] B) →ₗ[R] (⨁ (p : ιA × ιB), (𝒜 ⊗ ℬ) p) :=
TensorProduct.lift
  (DirectSum.toModule _ _ _ fun i ↦
    { toFun a := DirectSum.toModule _ _ _ fun j ↦
      { toFun b := DirectSum.lof R (ιA × ιB) (fun p => (𝒜 ⊗ ℬ) p) (i, j) ⟨a.1 ⊗ₜ b.1, ⟨a ⊗ₜ b, rfl⟩⟩
        map_add' := ...
        map_smul' := ... }
      map_add' := ...
      map_smul' := ...  }) 
    ∘ₗ
  TensorProduct.map 
    (DirectSum.decomposeLinearEquiv 𝒜 |>.toLinearMap) 
    (DirectSum.decomposeLinearEquiv ℬ |>.toLinearMap)
\end{lstlisting}
\begin{flushright}
\cite[file location: {\footnotesize\textsf{Grading/TensorProduct.lean}}]{projconstruction}
\end{flushright}
\end{construction}

\begin{theorem}{(Tensor product of relevant element)}\label{relevant-tensor-relevant}
  If $x$ is a relevant homogeneous element in $A$ and $y$ is a relevant homogeneous element in $B$, 
  $x \ox y$ is a relevant homogeneous element in $A \ox_R B$.
\begin{flushright}
\cite[\lstinline|TensorProduct.tmul_elemIsRelevant| file location: {\footnotesize\textsf{Grading/TensorProduct.lean}}]{projconstruction}
\end{flushright}
\end{theorem}
\begin{proof}{}
  Let us denote $M$ to be the subgroup $\iota_{A}\left[ \mybar{\left\langle x \right\rangle} \right]$ of $\iota_A$, 
  $N$ to be the subgroup $\iota_{B}\left[ \mybar{\left\langle y \right\rangle} \right]$ of $\iota_B$ and
  $X$ to be the subgroup $\left( \iota_A \times \iota_{B} \right)\left[ \mybar{\left\langle x\ox y \right\rangle} \right]$.
  We will show that $\quotient{\iota_A \times \iota_B}{X}$ is torsion.
  Let $e : \quotient{\iota_{A}}{M} \times \quotient{\iota_{B}}{N} \to \quotient{\iota_A\times\iota_B}{X}$ be the surjective group homomorphism induced by
  the two group homomorphisms $\iota_A \to \iota_A \times \iota_B \to \quotient{\iota_A \times \iota_B}{X}$ and 
  $\iota_B \to \iota_A \times \iota_B \to \quotient{\iota_A \times \iota_B}{X}$.
  Hence, we only need to show that $\quotient{\iota_A}{M} \times \quotient{\iota_{B}}{N}$ is torsion.
  Since $x$ is relevant, $\quotient{\iota_A}{M}$ is torsion, and since $y$ is relevant, $\quotient{\iota_B}{N}$ is torsion; consequently, their product is torsion.
  
\end{proof}

\begin{theorem}{}\label{relevant-elem-in-tensor-product}
  Suppose $\iota_A$ and $\iota_{B}$ are both finitely generated abelian groups.
  For any relevant homogeneous element $x \in A \ox_R B$, 
  there exists $s^{A}_1, \dots, s^{A}_n$ in $A$ and $s^{B}_1, \dots, s^{B}_n$ in $B$ such that each $s^{A}_i$ and $s^{B}_j$ are relevant homogeneous elements 
  and for some natural number $k$, 
  \[
  x^k = \sum_{i=1}^n s^{A}_i \ox s^{B}_i.
  \]
\begin{flushright}
\cite[\lstinline|TensorProduct.elemIsRelevant_of_exists| file location: {\footnotesize\textsf{Grading/TensorProduct.lean}}]{projconstruction}
\end{flushright}
\end{theorem}
\begin{proof}{}
  Since $x$ is relevant, by~\cref{elemIsRelevant-iff}, 
    we can find elements $y_1, \dots, y_N$ in $A \ox_R B$ such that each $y_i$ has degree 
      $\left(d^{A}_i, d^{B}_i\right) \in \iota_{A} \times \iota_{B}$ where $\left\langle \left(d^{A}_i, d^{B}_i\right)\middle| i = 1,\dots,N \right\rangle$
      is a subgroup of finite index in $\iota_{A} \times \iota_{B}$ and $x^k = \prod_i y_i$.
    Hence, $\left\langle d^{A}_{i} \middle| i = 1,\dots, N \right\rangle$ and $\left\langle d^{B}_{i} \middle| i = 1,\dots, N \right\rangle$ are subgroups of finite index in $\iota_A$ and $\iota_B$ respectively.
  Since each $y_i$ has degree $\left(d^{A}_i, d^{B}_i\right)$, we can write each $y_i$ as $\sum_{j \in J_i} s^{A}_{j} \ox s^{B}_j$ where $J_i$ is
  a finite set of indices and for any $j \in J_i$, $s^{A}_{j}$ and $s^{B}_{j}$ as degree $d^{A}_i$ and $d^{B}_i$ respective.
  Hence,
  \[
  \begin{aligned}
  x^k &= \prod_{i=1}^N y_i \\
  &= \prod_{i=1}^N \left( \sum_{j \in J_i} s^{A}_{j} \ox s^{B}_j \right) \\
  &= \sum_{\left(j_1, \dots, j_N \middle| j_i \in J_i \right)} \left(\prod_{i = 1}^N s^{A}_{j_i} \right) \ox \left(\prod_{i = 1}^N s^{B}_{j_i} \right)
  \end{aligned}.
  \]
  By~\cref{elemIsRelevant-iff} again, all the products $\prod_{i = 1}^N s^{A}_{j_i}$ and $\prod_{i = 1}^N s^{B}_{j_i}$ are relevant homogeneous elements. Thus, the theorem holds.
\end{proof}

\begin{theorem}{}{}
  The range of the linear map $A_{\dagger} \ox_R B_{\dagger} \to A \ox_R B$ is an ideal of $A \ox_R B$, and we have
  \[
  \rad{\left( A\ox_R B \right)_{\dagger}} = \rad{\im\left( A_{\dagger} \ox_R B_{\dagger} \to A \ox_R B \right)}.
  \]
\begin{flushright}
\cite[\lstinline|TensorProduct.rad_dagger| file location: {\footnotesize\textsf{Grading/TensorProduct.lean}}]{projconstruction}
\end{flushright}
\end{theorem}
\begin{proof}{}
  We first check that the range of the linear map $A_{\dagger} \ox_R B_{\dagger} \to A \ox_R B$ is an ideal of $A \ox_R B$.
  Suppose $a \ox b \in A_{\dagger} \ox_R B_{\dagger}$ and $a' \ox b' \in A \ox_R B$ is in $\im\left( A_{\dagger} \ox_R B_{\dagger} \to A \ox_R B \right)$. Since $a \in A_{\dagger}$, $aa'$ is in $A_{\dagger}$, and since $b \in B_{\dagger}$, $bb'$ is in $B_{\dagger}$.
  Therefore, $(a' \ox b')(a \ox b)$ is in the range of the linear map $A_{\dagger} \ox_R B_{\dagger} \to A \ox_R B$.

  We then demonstrate $\rad{\left( A\ox_R B \right)_{\dagger}} \le \rad{\im\left( A_{\dagger} \ox_R B_{\dagger} \to A \ox_R B \right)}$.
  It is sufficient to show that every relevant homogeneous element $x \in A \ox_R B$ is in the range. 
  By~\cref{relevant-elem-in-tensor-product}, we can find relevant homogeneous elements $s^{A}_1, \dots, s^{A}_n$ in $A$ and $s^{B}_1, \dots, s^{B}_n$ in $B$ such that $x^k = \sum s_{i}^{A} \ox s_{i}^{B}$ for some natural number $k$. 
  Since each $s^{A}_{i}$ is in $A_{\dagger}$ and each $s^{B}_{i}$ is in $B_{\dagger}$, we have $\sum s^{A}_{i} \ox s^{B}_{i}$ is in $\im\left( A_{\dagger} \ox_R B_{\dagger} \to A \ox_R B \right)$ and consequently, $x$ is the radical.

  For the other direction, we demonstrate $\im\left( A_{\dagger} \ox_R B_{\dagger} \to A \ox_R B \right) \le \left( A\ox_R B \right)_{\dagger}$.
  Since $\left( A\ox_R B \right)_{\dagger}$ is an $R$-submodule of $A \ox_R B$ as well, it is sufficient to show that 
  for every element $a \ox b \in A_{\dagger} \ox_R B_{\dagger}$, $a \ox b$ is in $\left( A\ox_R B \right)_{\dagger}$.  
  This is precisely~\cref{relevant-tensor-relevant}.
\end{proof}

\section{Potions}\label{sec:potions}
In this section, we assume $A$ is a commutative $R_0$-algebra that is graded by an abelian group $\iota$. 
We use the same terminology found in~\cite{MR24}.

\begin{definition}{}{} \label{defi:potion}
  For a homogeneous submonoid $S \subseteq A$, the homogeneous localization $\hloc A S$ is called the \emph{potion} of $A$ with respect to $S$.
\begin{flushright}
\cite[\lstinline|Potion| file location: {\footnotesize\textsf{Potions/Basic.lean}}]{projconstruction}
\end{flushright}
\end{definition}
Since we will be working with multiple submonoids at once, the notation $\hloc A S$ can be confusing --- for example 
$\hloc A {(S_1 S_2)S_3}$ and $\hloc A {S_1(S_2S_3)}$ could be potentially misread to mean (normal) localization at $(S_1S_2)S_3$ with a redundant bracket in the notation.
Therefore, we use $\potion(S)$ to denote the ring $\hloc A S$. 
Type theoretically, since $S$ is a homogeneous submonoid of $A$, the absence of $A$ in the notation $\potion(S)$ should cause no issue.
If the ring of which $S$ is a submonoid could not be inferred from the context, we will use $\potion_A(S)$ instead.

We first collect some useful morphisms between potions defined by Lemma \ref{hloc-map}
\begin{itemize}
  \item If $\phi : A \to B$ is a graded ring homomorphism between two $\iota$-graded rings, there is a ring homomorphism
  $\leanmath{potionToMap}: \potion(S) \to \potion\left(\phi_{\star} S\right)$ induced by $\phi$.
  \item If $S = T$ are two homogeneous submonoids of $A$, we have a ring isomorphism $\leanmath{potionEquiv}: \potion(S) \cong \potion(T)$ induced by the $\id_A$.
  In particular, since $S S = S$, we have a ring isomorphism $\potion(S)\cong \potion(SS)$.
  \item For any two homogeneous submonoids $S$ and $T$, there is a ring homomorphism $\leanmath{potionToMul}:\potion(S) \to \potion(ST)$ induced by $\id_A$ and the following commutative square:
  \[
  \begin{tikzcd}[column sep=huge]
  \potion(S) \arrow[r, "\leanmath{potionToMul}"] \arrow[d, "\leanmath{potionToMap}"] 
  & \potion(ST) \arrow[r, "\leanmath{potionToMap}"] 
  & \potion\left(\phi_{\star}(S T)\right) \arrow[d, "\leanmath{potionEquiv}"']\\
  \potion\left(\phi_{\star} S\right) \arrow[rr, "\leanmath{potionToMul}"] & 
  & \potion\left(\left( \phi_{\star}S \right)\left( \phi_{\star}T \right)\right)
  \end{tikzcd}.
  \]
  With the ring homomorphism $\leanmath{potionToMul}$, we view $\potion(ST)$ as a $\potion(S)$-algebra.
  \item There is a ring homomorphism $\leanmath{toBarPotion}: \potion(S) \to \potion\left( \mybar{S} \right)$ induced by $\id_A$.
  Since $\leanmath{toBarPotion}$ is bijective, we have a ring isomorphism $\leanmath{equivBarPotion}: \potion(S) \cong \potion\left( \mybar{S} \right)$.
  For any homogeneous elements $a, b\in A$ of degree $i$ and $c \in A$ of degree $j$ such that $bc \in S$, 
  the image of the homogeneous fraction $\frac{m}{n} \in \potion\left( \mybar{S} \right)$ under $\leanmath{equivBarPotion}^{-1}$ is the homogeneous fraction
  $\frac{az}{bz} \in \potion(S)$.
\end{itemize}

\begin{definition}{}{}
  Suppose $S$ and $T$ be homogeneous submonoids of $A$. A \emph{potion generator} of $T$ over $S$ is the following data:
\begin{itemize}
  \item \sethlcolor{yellow!10}\hl{an indexing set $I$};
  \item \sethlcolor{red!10}\hl{a family of elements $\left\{ t_i \middle| i \in I \right\} \subseteq T$ generating $T$ as submonoid};
  \item \sethlcolor{blue!10}\hl{two families of elements $\left\{ s_i \middle| i \in I \right\} \subseteq \mybar{S}$ and  
  $\left\{ s_i' \middle| i \in I \right\} \subseteq \mybar{S}$ such that each $s_j$ is homogeneous of degree $i_j$ and 
  each $s_j'$ is homogeneous of degree $i'_j$};
  \item \sethlcolor{green!10}\hl{a family of positive natural numbers $\left\{ n_i\middle| i\in I \right\}$, such that for each $j \in I$,
  $t_j^{n_j}$ is homogeneous of degree $i_j - i'_j$ (the two families of indices $i_j,i'_j$ are also part of the data).}
\end{itemize}
For each $i$, the fraction $\frac{t_i^{n_i}s'_i}{s_i}$ is in $\potion\left( \mybar{S} \right)$.
We denote $\mathsf{s}(T')$ to \sethlcolor{olive!10}\hl{the submonoid of $S$ generated by the $f_i$ where $f_i\in\potion(S)$ 
corresponds to the fraction $\frac{t_i^{n_i}s'_i}{s_i} \in \potion\left( \mybar{S} \right)$ under $\leanmath{equivBarPotion}$}.

\begin{lstlisting}[extendedchars=true, caption={Potion Generator}, 
  linebackgroundcolor={%
  \ifnum\value{lstnumber}=2\color{yellow!10}\fi
  \ifnum\value{lstnumber}=3\color{red!10}\fi
  \ifnum\value{lstnumber}=4\color{red!10}\fi
  \ifnum\value{lstnumber}=5\color{red!10}\fi
  \ifnum\value{lstnumber}=6\color{blue!10}\fi
  \ifnum\value{lstnumber}=7\color{blue!10}\fi
  \ifnum\value{lstnumber}=8\color{blue!10}\fi
  \ifnum\value{lstnumber}=9\color{blue!10}\fi
  \ifnum\value{lstnumber}=10\color{blue!10}\fi
  \ifnum\value{lstnumber}=11\color{blue!10}\fi
  \ifnum\value{lstnumber}=12\color{green!10}\fi
  \ifnum\value{lstnumber}=13\color{green!10}\fi
  \ifnum\value{lstnumber}=16\color{olive!10}\fi
  \ifnum\value{lstnumber}=17\color{olive!10}\fi
  \ifnum\value{lstnumber}=18\color{olive!10}\fi
  \ifnum\value{lstnumber}=19\color{olive!10}\fi
  \ifnum\value{lstnumber}=20\color{olive!10}\fi
  \ifnum\value{lstnumber}=21\color{olive!10}\fi
  \ifnum\value{lstnumber}=22\color{olive!10}\fi}, linebackgroundsep=-8pt, linebackgroundwidth=280pt]
structure PotionGen where
  (index : Type*)
  (elem : index → A)
  (elem_mem : ∀ t, elem t ∈ T)
  (gen : Submonoid.closure (Set.range elem) = T.toSubmonoid)
  (s s' : index → A)
  (s_mem_bar : ∀ t, s t ∈ S.bar)
  (s'_mem_bar : ∀ t, s' t ∈ S.bar)
  (i i' : index → ι)
  (s_deg : ∀ t, s t ∈ 𝒜 (i t))
  (s'_deg : ∀ t, s' t ∈ 𝒜 (i' t))
  (n : index → ℕ+)
  (t_deg : ∀ t : index, (elem t : A)^(n t : ℕ) ∈ 𝒜 (i t - i' t))

def PotionGen.genSubmonoid (T' : PotionGen S T) : Submonoid S.Potion :=
  Submonoid.closure
    { x | ∃ (t : T'.index), x =
      S.equivBarPotion.symm (.mk
        { deg := T'.i t,
          num := ⟨(T'.elem t) ^ (T'.n t : ℕ) * T'.s' t, ...⟩,
          den := ⟨T'.s t, T'.s_deg t⟩,
          den_mem := T'.s_mem_bar t }) }
\end{lstlisting}
When the indexing set $I$ is finite, we say that $T'$ is a \emph{finite potion generator} of $T$ over $S$.
\begin{flushright}
\cite[\lstinline|PotionGen| file location: {\footnotesize\textsf{Potions/Basic.lean}}]{projconstruction}
\end{flushright}
\end{definition}

\begin{remark}{}\label{potion-gen-disj-union}
  For any three $R, S, T$ homogeneous submonoids of $A$, 
  if $R' = \left( I_R, t_R, s_R, {s'}_R, i_R, {i'}_R, n_R \right)$ is a potion generator of $R$ over $S$ and 
  $T' = \left( I_T, t_T, s_T, {s'}_T, i_T, {i'}_T, n_T \right)$ is a potion generator of $T$ over $S$, 
  we have a potion generator of $RT$ over $S$ given by the disjoint union of the two potion generators:
  \[R' \oplus T' = (I_R \oplus I_T, s_R \oplus s_T, {s'}_R \oplus {s'}_T, i_R \oplus i_T, {i'}_R \oplus {i'}_T, n_R \oplus n_T).\]
It is useful to note that $\sfs\left( R'\oplus T' \right)$ is equal to $\sfs(R')\sfs(T')$.
\begin{lstlisting}[caption={Disjoint union of potion generators}]
def PotionGen.disjUnion {R S T : HomogeneousSubmonoid 𝒜} (R' : PotionGen S R) (T' : PotionGen S T) :
    PotionGen S (R * T) where
  index := R'.index ⊕ T'.index
  elem := Sum.rec R'.elem T'.elem
  elem_mem := ...
  gen := ...
  n := Sum.rec R'.n T'.n
  s := Sum.rec R'.s T'.s
  s' := Sum.rec R'.s' T'.s'
  s_mem_bar := Sum.rec R'.s_mem_bar T'.s_mem_bar
  s'_mem_bar := Sum.rec R'.s'_mem_bar T'.s'_mem_bar
  i := Sum.rec R'.i T'.i
  i' := Sum.rec R'.i' T'.i'
  t_deg := Sum.rec R'.t_deg T'.t_deg
  s_deg := Sum.rec R'.s_deg T'.s_deg
  s'_deg := Sum.rec R'.s'_deg T'.s'_deg

lemma PotionGen.disjUnion_genSubmonoid {R S T : HomogeneousSubmonoid 𝒜}
      (R' : PotionGen S R) (T' : PotionGen S T) :
    (R'.disjUnion T').genSubmonoid = R'.genSubmonoid * T'.genSubmonoid := ...
\end{lstlisting}
\begin{flushright}
\cite[\lstinline|PotionGen.disjUnion| file location: {\footnotesize\textsf{Potions/Basic.lean}}]{projconstruction}
\end{flushright}
\end{remark}

The motivation for the definition of potion generator is the following theorem.
\begin{theorem}{}\label{potion-loc-iso}
  Let $T' = \left( I, t, s, s', i, i', n \right)$ be a potion generator of $T$ over $S$.
  We have a $\potion(S)$-algebra isomorphism between $\potion( ST )$ and $\potion(S)_{\sfs(T')}$.
\begin{flushright}
\cite[\lstinline|localizationAlgEquivPotion| file location: {\footnotesize\textsf{Potions/Localization.lean}}]{projconstruction}
\end{flushright}
\end{theorem}
\begin{proof}{}
  In order to descend the ring homomorphism $\leanmath{potionToMul} : \potion(S) \to \potion(ST)$ to a ring homomorphism $\potion(S)_{\sfs(T')} \to \potion(ST)$,
  we need to show that the image of $\leanmath{equivBarPotion}^{-1}\left( \frac{t_i^{n_i}s'_i}{s_i} \right)$ is sent to a unit in $\potion(ST)$.
  Since $s_i$ and $s'_i$ are in $\mybar{S}$, we can find $y$ and $y'$ in $S$ such that $s_i\mid y$ and $s_i' \mid y'$.
  Hence, by~\cref{exists-homogeneous-of-dvd}, there exists two homogeneous elements $z$ and $z'$ such that $y = s_i z$ and $y' = s_i' z'$.
  Therefore, when seen as a fraction in $\potion(S)$, $\frac{t_i^{n_i}s'_i}{s_i}$ is equal to $\frac{t_i^{n_i}s'_iz}{s_i z}$ and, 
  its image under $\leanmath{potionToMul}$ is $\frac{t_i^{n_i}s'_iz}{s_i z}$ with inverse $\frac{s_i z'}{t_i^{n_i}s'_iz'}$.
  Therefore, we have a well-defined ring homomorphism $\phi: \potion(S)_{\sfs(T')} \to \potion(ST)$.
\begin{lstlisting}[extendedchars=true]
def localizationToPotion (T' : PotionGen S T) :
    Localization T'.genSubmonoid →+* (S * T).Potion :=
  @IsLocalization.lift
    (R := S.Potion)
    (M :=  _)
    (S :=  _)
    (P := (S * T).Potion)
    (g := S.potionToMul T) _ _ _ _
    (Localization.isLocalization (R := S.Potion) (M := T'.genSubmonoid)) 
    ...
\end{lstlisting}
An element in $\potion(S)_{\sfs(T')}$ has the form $\frac{\sfrac{a}{s}}{\prod_i\leanmath{equivBarPotion}^{-1}\left(\sfrac{t_i^{n_i}s'_i}{s_i}\right)^{k_i}}$ where 
in the denominator, the indices $i$ runs through a finite subset of $I$ and $k_i$ are positive natural numbers. 
With this notation, we see that the image under $\phi$ is equal to $\frac{a}{s}\cdot \prod_i \left( \frac{s_i}{t_i^{n_i}s'_i} \right)^{k_i}$.
Since, $\prod_i \left( \frac{s_i}{t_i^{n_i}s'_i} \right)^{k_i}$ is invertible in $\potion(ST)$, the image is zero if and only the numerator $\frac{a}{s} \in \potion(ST)$ is zero.

We first show that $\phi$ is injective. Suppose an element $\frac{\sfrac{a}{s}}{...}$ is sent to zero under $\phi$. 
We want to show that $\frac{\sfrac{a}{s}}{...}$ is equal to zero.
It is sufficient to find an element $x$ in $\sfs(T')$ such that $x \cdot \frac{a}{s}$ is zero as elements of $\potion(S)$.
Since $\phi\left(\frac{\sfrac{a}{s}}{\dots}\right)$ is zero, we see that $\frac{a}{s}$ is equal to zero in $\potion(ST)$, therefore,
there exists an element $\mathfrak{s} \in S$ and $\mathfrak{t} \in T$ such that $\mathfrak{s} \mathfrak{t} a = 0$.
Since $\left\{ t_i | i \in I \right\}$ generates $T$, we can write $\mathfrak{t}$ as $\prod_i {t_i}^{k_i}$. 
Set $x = \prod_i \leanmath{equivBarPotion}^{-1}\left( \frac{\mathfrak{s}t_i^{n_i}s'_i}{\mathfrak{s}s_i} \right)^{k_i}$, we see that $x\cdot \frac a s$ is equal to zero.

Then, we show that $\phi$ is surjective. Suppose $\frac{a}{\mathfrak{s}\mathfrak{t}}$ is an element in $\potion(ST)$ and we write $\mathfrak{t}$ as $\prod_i t_i^{k_i}$. 
We have the following equality in $\potion(ST)$:
\[
\frac{a}{\mathfrak{s}\mathfrak{t}} = \frac{a}{\mathfrak{s}\prod_i t_i^{k_i}} = \frac{a \prod_{i}t_i^{k_i (n_i - 1){s'_{i}}^{k_i}}}{\mathfrak{s}\prod_i{s_i}^{k_i}}\cdot \prod_i\left(\frac{s_i}{t_i^{n_i}{s'_i}}\right)^{k_i}.
\]
Hence, $\frac{a}{\mathfrak{s}\mathfrak{t}}$ is in the image of $\phi$.
Thus, $\phi$ is a ring isomorphism between $\potion(S)_{\sfs(T')}$ and $\potion(ST)$. One can check that $\phi$ is a $\potion(S)$-algebra isomorphism as well.
\begin{lstlisting}[extendedchars=true]
def localizationRingEquivPotion (T' : PotionGen S T) :
    Localization T'.genSubmonoid ≃+* (S * T).Potion :=
  RingEquiv.ofBijective (localizationToPotion T') ...

def localizationAlgEquivPotion (T' : PotionGen S T) :
    Localization T'.genSubmonoid ≃ₐ[S.Potion] (S * T).Potion :=
  AlgEquiv.ofRingEquiv (f := localizationRingEquivPotion T') fun x ↦ by
    induction x using Quotient.inductionOn' with | h x =>
    simp [localizationToPotion, Localization.mk_eq_mk', IsLocalization.lift_mk']
\end{lstlisting}
\end{proof}

\begin{corollary}\label{spec-potionToMul-is-open-immersion}
  Suppose there exists a finite potion generator $T' $ of $T$ over $S$.
  The morphism of scheme $\spec \potion(ST) \to \spec \potion(S)$ is an open immersion.
\end{corollary}
\begin{proof}{}
  We have the following commutative triangle:
  \[
  \begin{tikzcd}[column sep=large]
    \potion(S) \arrow[r, "\leanmath{potionToMul}"] \arrow[d, "\frac{\bullet}{1}"] & \potion(ST) \arrow[ld, "\cong"] \\
    \potion(S)_{\sfs(T')} & \\
  \end{tikzcd}.
  \]
  Therefore, the morphism $\spec \potion(ST) \to \spec \potion(S)$ factors as
  \[\spec \potion(ST) \cong \spec \potion(S)_{\sfs(T')}  \to  \spec \potion(S).\]
  Hence, it is sufficient to show that $\spec \potion(S)_{\sfs(T')} \to  \spec \potion(S) $ is an open immersion.
  This is true because $\sfs(T')$ is generated by a finite set.
\end{proof}

\begin{theorem}\label{exists-finite-potion-generator}
  If $S$ is relevant and $T$ is finitely generated as a submonoid of $A$, there exists a finite potion generator $T'$ of $T$ over $S$. 
\begin{flushright}
\cite[\lstinline|finitePotionGen| file location: {\footnotesize\textsf{Potions/Basic.lean}}]{projconstruction}
\end{flushright}
\end{theorem}
\begin{proof}{}
  We claim that for every homogeneous element $t \in A$, there exists a positive natural number $n$,
  and two homogeneous elements $s, s' \in \mybar{S}$ of degree $i$ and $i'$ respectively such that $t^n \in A_{i - i'}$:
\begin{lstlisting}[caption={Exsitence of finite potion generator}]
lemma finite_potionGen_exists_aux₂ (S_rel : IsRelevant S) (t : A) (ht : SetLike.Homogeneous 𝒜 t) :
  ∃ (n : ℕ+) (s s' : A) (i i' : ι),
    t^(n : ℕ) ∈ 𝒜 (i - i') ∧ s ∈ 𝒜 i ∧ s' ∈ 𝒜 i' ∧ s ∈ S.bar ∧ s' ∈ S.bar := ...
\end{lstlisting}
  The claim is true: let $t$ be a homogeneous element of degree $m$.  Since $S$ is relevant, there exists a positive natural number $n$ such that $t^n \in \iota\left[ \mybar{S} \right]$, 
  that is, there exists $i$ and $i'$ in $\iota$ such that $n\cdot m = i - i'$ and $i, i'$ are elements of $\deg\left( \mybar{S} \right)$.
  Hence, there exists $s$ and $s'$ in $\mybar{S}$ with degree $i$ and $i'$ respectively.
  
Since $T$ is finitely generated, we can choose an arbitrary finite set $T' \subseteq T$ which generates $T$ as a submonoid.
  Hence, we can define functions 
  $n : T' \to \NN_{>0}$, $s, s' : T' \to A$, and $i, i' : T' \to \iota$ such that for each $t \in T'$,
  $s(t)$ and $s'(t)$ are homogeneous elements of degree $i(t)$ and $i'(t)$ respectively, and $t^{n(t)}$ is homogeneous of degree $i(t) - i'(t)$.
  These are exactly the data we need to define a finite potion generator of $T$ over $S$.
  \begin{lstlisting}[extendedchars=true]
def finitePotionGen (S_rel : IsRelevant S) (T_fg : T.FG) : PotionGen S T :=
  let carrier := T_fg.choose
  let gen : Submonoid.closure carrier = T.toSubmonoid := T_fg.choose_spec
  let n : carrier → ℕ+ := fun t ↦ (finite_potionGen_exists_aux₂ S_rel t ...).choose
  let s : carrier → A :=
    fun t ↦ (finite_potionGen_exists_aux₂ S_rel t ...).choose_spec.choose
  let s' : carrier → A := fun t ↦
    (finite_potionGen_exists_aux₂ S_rel t ...).choose_spec.choose_spec.choose
  let i : carrier → ι := fun t ↦
    (finite_potionGen_exists_aux₂ S_rel t ...).choose_spec.choose_spec.choose_spec.choose
  let i' : carrier → ι := fun t ↦
    (finite_potionGen_exists_aux₂ S_rel t ...).choose_spec.choose_spec.choose_spec.choose_spec.choose
  ...
  { index := carrier
    elem := Subtype.val
    n := n
    s := s
    s' := s'
    i := i
    i' := i'
    elem_mem := ...
    ... }
  
lemma finitePotionGen_finite (S_rel : IsRelevant S) (T_fg : T.FG)  :
  Finite (finitePotionGen S_rel T_fg).index := T_fg.choose.finite_toSet
  \end{lstlisting}
\end{proof}

\begin{corollary}\label{potionToMul-open-immersion}
  If $S$ is relevant and $T$ is finitely generated as a submonoid of $A$, then $\spec \potion(ST) \to \spec \potion(S)$ is an open immersion.
\begin{lstlisting}
theorem IsOpenImmersion.of_isRelevant_FG (S_rel : IsRelevant S) (T_fg : T.FG) :
    IsOpenImmersion <| Spec.map <| CommRingCat.ofHom (S.potionToMul T) := ...
\end{lstlisting} 
\begin{flushright}
\cite[\lstinline|IsOpenImmersion.of_isRelevant_FG| file location: {\footnotesize\textsf{Potions/Localization.lean}}]{projconstruction}
\end{flushright}
\end{corollary}

\section{Gluing Potions and defining Proj}\label{sec:glu-proj}
Now, we restrict our attention to relevant and finitely generated homogeneous submonoids; we call such submonoids \emph{good potion ingredients}.
\begin{lstlisting}[caption={Good potion ingredients}]
structure GoodPotionIngredient extends (HomogeneousSubmonoid 𝒜) where
  relevant : toHomogeneousSubmonoid.IsRelevant
  fg : toSubmonoid.FG
\end{lstlisting}
We immediately notice that product of good potion ingredients is another good potion ingredient, 
hence the set of good potion ingredients is a commutative semigroup. 
If $\phi : A \to B$ is a graded ring homomorphism, and $S$ is a good potion ingredient of $A$, 
then $\phi_{\star} S$ is a good potion ingredient of $B$.
For any three good potion ingredients $R, S, T$, we view $\potion(RST)$ as a $\potion(S)$-algebra via the ring homomorphism
$\potion(S) \to \potion(S(RT)) \to \potion(RST)$.

Let $R'$ be a potion generator of $R$ over $S$ and $T'$ be a potion generator of $T$ over $S$. 
\begin{construction}\label{construction-t-prime}
  We have a $\potion(S)$-algebra isomorphism $e: \potion(ST) \ox_R \potion(SR) \cong \potion(RST)$ by 
  composing the isomorphisms in ~\cref{potion-mixing-iso}. 
  \begin{lstlisting}[caption={Mixing good potion ingredients}]
def mixing : 
    (S * T).Potion ⊗[S.Potion] (S * R).Potion ≃ₐ[S.Potion] (R * S * T).Potion :=
  mixingAux₀ R' T' |>.trans <| mixingAux₁ R' T' |>.trans <| mixingAux₂ R' T' |>.trans <|
  mixingAux₃ R' T' |>.trans <| mixingAux₄ R S T
  \end{lstlisting}
We notice that for any $x \in \potion(ST)$, $e(x\ox 1)$ is equal to the image of $x$ under $\potion(ST) \to \potion(STR) \to \potion(RST)$ and
for any $x \in \potion(SR)$, $e(1\ox x)$ is equal to the image of $x$ under $\potion(SR)\to \potion(SRT) \to \potion(RST)$.
\begin{lstlisting}[extendedchars=true]
lemma mixing_left (x : (S * T).Potion) :
    mixing R' T' (x ⊗ₜ 1) = potionEquiv ... (potionToMul _ R.1 x) := ...

lemma mixing_right (x : (S * R).Potion) :
    mixing R' T' (1 ⊗ₜ x) = potionEquiv ... (potionToMul _ T.1 x) := ...
\end{lstlisting}
With the same construction, we have an isomorphism $\potion(RS) \ox_{\potion(R)} \potion(RT) \cong \potion(TRS) \cong \potion(RST)$.
Hence, we have an isomorphism\footnote{this isomorphism motivates the name ``good potion ingredients'', because the potions made from $R$, $S$ and $T$ mix well.} 
$t'_{RST}: \potion(ST)\ox_{\potion(S)}\potion(SR) \cong \potion(RS) \ox_{\potion(R)} \potion(RT)$.

\begin{lstlisting}[extendedchars=true]
def t'Aux₀ (R S T : GoodPotionIngredient 𝒜) :
    (S * T).Potion ⊗[S.Potion] (S * R).Potion ≃+* (R * S * T).Potion :=
  mixing (finitePotionGen S.relevant R.fg) (finitePotionGen S.relevant T.fg)

def t'Aux₁ (R S T : GoodPotionIngredient 𝒜) :
    (R * S).Potion ⊗[R.Potion] (R * T).Potion ≃+* (R * S * T).Potion :=
  (mixing (finitePotionGen R.relevant T.fg) (finitePotionGen R.relevant S.fg)).toRingEquiv.trans <|
    potionEquiv (by rw [mul_comm T, mul_assoc, mul_comm T, ← mul_assoc])

def t' (R S T : GoodPotionIngredient 𝒜) :
    ((S * T).Potion ⊗[S.Potion] (S * R).Potion) ≃+*
    ((R * S).Potion ⊗[R.Potion] (R * T).Potion) :=
  (t'Aux₀ R S T).trans (t'Aux₁ R S T).symm
\end{lstlisting} 
\begin{flushright}
\cite[file location: {\footnotesize\textsf{Potions/GoodPotionIngredient.lean}}]{projconstruction}
\end{flushright}
\end{construction}

\begin{corollary}\label{t-prime-cocyle-condtion-ring-version}
For any three good potion ingredients $R, S, T$, $t'_{TRS} \circ t'_{STR}\circ t'_{RST}$ is the identity isomorphism.
\begin{lstlisting}
lemma t'_cocycle (R S T : GoodPotionIngredient 𝒜) :
    (T.t' R S).trans ((S.t' T R).trans (R.t' S T))  = RingEquiv.refl _ := 
  ...
\end{lstlisting}
\end{corollary}

\begin{corollary}\label{t-prime-fac-ring-version}
For any three good potion ingredients $R, S, T$, the following diagram commutes:
\[
\begin{tikzcd}
\potion(SR) \arrow[r, "1\ox\bullet"] \arrow[rd, "\cong", "\leanmath{potionEquiv}"'] & 
\potion(ST) \ox_{\potion(S)} \potion(SR) \arrow[r, "t'_{RST}"] & \potion(RS) \ox_{\potion(R)} \potion(RT) \\
& \potion(RS) \arrow[ru, "1\ox\bullet"]& 
\end{tikzcd}.
\]
\begin{lstlisting}
lemma t'_fac (R S T : GoodPotionIngredient 𝒜) :
    ((R.t' S T)).toRingHom.comp Algebra.TensorProduct.includeRight.toRingHom =
    Algebra.TensorProduct.includeLeftRingHom.comp
      (potionEquiv <| by rw [mul_comm]).toRingHom := 
  ...
\end{lstlisting}
\end{corollary}

Now we are ready to proceed with $\proj$ construction.
Let $\mathscr{F} = \left\{ S_i \middle| i \in \tau \right\}$ be a family of good potion ingredients of $A$ indexed by $\tau$. We aim to glue the family of schemes
together $\left\{ \spec \potion\left( S_i \right) \middle| i \in \tau \right\}$.
\begin{construction}{}{} \label{const:proj}
To proceed with the gluing, we need the following data:
\begin{itemize}
  \item an indexing type $J$;
  \item a scheme ${U_i}$ for each ${i} \in J$;
  \item a scheme ${V_{ij}}$ for each pair ${(i,j)\in J\times J}$ representing the intersection of ${U_i}$ and ${U_j}$;
  \item an open immersion of schemes $f_{ij} : V_{ij} \to U_i$ for each pair ${(i,j)\in J\times J}$ such that $f_{ii}$ is an isomorphism;
  \item a morphism of schemes $t'_{ijk} : V_{ij} \times_{U_i} V_{ik} \to V_{jk} \times_{U_j} V_{ji}$ such that for each triple ${(i,j,k)\in J\times J\times J}$,
  the composition
  \begin{equation}
    \label{t-prime-cocycle-condition}
    \begin{tikzcd}
      V_{ij} \times_{U_i} V_{ik} \arrow[r, "t'_{ijk}"]&
      V_{jk} \times_{U_j} V_{ji} \arrow[r, "t'_{jki}"]& 
      V_{ki} \times_{U_k} V_{kj} \arrow[r, "t'_{kij}"]& 
      V_{ij} \times_{U_i} V_{ik} 
    \end{tikzcd}
  \end{equation}
  is the identity map;
  \item a transition map $t_{ij} : V_{ij} \to V_{ji}$ for each pair ${(i,j)\in J\times J}$ such that $t_{ii}$ is the identity map, and
  for each triple ${(i,j,k)\in J\times J\times J}$, the following diagram commutes:
  \begin{equation}
    \label{t-fac-condition}
    \begin{tikzcd}
      V_{ij} \times_{U_i} V_{ik} \arrow[r, "p_1"] \arrow[rd, "t'_{ijk}"'] & V_{ij} \arrow[r, "t_{ij}"] & V_{ji}\\
      & V_{jk} \times_{U_j} V_{ji} \arrow[ru, "p_2"'] &
    \end{tikzcd}.
  \end{equation}
  
\end{itemize}
In our case, we will take $J = \tau$, $U_i = \spec\left(\potion(S_i)\right)$, $V_{ij} = \spec\left(\potion(S_i S_j)\right)$, and
$f_{ij} = \spec \left(\potion(S_i) \to \potion\left( S_i S_j \right)\right)$, $t_{ij} = \spec\left( \potion\left( S_i S_j \right) \to \potion\left( S_j S_i \right) \right)$
and $t'_{ijk}$ to be the composition
\[
\begin{tikzcd}
\spec\left( \potion\left( S_i S_j \right) \right) \times_{\spec\left( S_i \right)} \spec\left( \potion\left( S_j S_k \right) \right) \arrow[r, "\cong"] &
\spec\left( \potion\left( S_i S_j \right) \ox_{\potion\left( S_i \right)} \potion\left( S_j S_k \right)\right) \arrow[d, "\spec\left( t'_{S_i S_j S_k} \right)", "\cong"']& \\
\spec\left( \potion\left( S_j S_k \right) \right) \times_{\spec\left( S_j \right)} \spec\left( \potion\left( S_j S_i \right) \right) &
\spec\left( \potion\left( S_j S_k \right) \ox_{\potion\left( S_j \right)} \potion\left( S_j S_i \right)\right) \arrow[l, "\cong"']
\end{tikzcd},
\]
where $t'_{S_iS_jS_k}$ is defined in ~\cref{construction-t-prime}.

Since for any homogeneous submonoid $S$, $SS = S$, we see that $f_{ii}$ is the identity morphism. 
We have already shown that $f_{ij}$ are open immersions in~\cref{spec-potionToMul-is-open-immersion}.
Since $t_{ii}$ is defined as $\spec\left( \potion\left(S_i S_i\right) \to \potion\left( S_i S_i \right) \right)$, it is certainly the identity morphism.
Now, we verify \cref{t-prime-cocycle-condition} and \cref{t-fac-condition}:
modulo the isomorphisms of the form $\spec A \times_{\spec B} \spec C \cong \spec \left( A\ox_B C \right)$, 
\cref{t-prime-cocycle-condition} and \cref{t-fac-condition} are $\spec$ applied to~\cref{t-prime-cocyle-condtion-ring-version} and~\cref{t-prime-fac-ring-version} respectively.

\begin{lstlisting}[extendedchars=true]
def glueData {τ : Type u} (ℱ : τ → GoodPotionIngredient 𝒜) : Scheme.GlueData where
  J := τ
  U i := Spec <| CommRingCat.of <| (ℱ i).Potion
  V pair := Spec <| CommRingCat.of <| (ℱ pair.1 * ℱ pair.2).Potion
  f i j := Spec.map <| CommRingCat.ofHom <| (ℱ i).potionToMul (ℱ j).toHomogeneousSubmonoid
  f_id i := ...
  f_open i j := isOpenImmersion (ℱ i) (ℱ j)
  t i j := Spec.map <| CommRingCat.ofHom <|  potionEquiv (mul_comm ..) |>.toRingHom
  t_id i := by
    erw [← Scheme.Spec.map_id]
    simp
  t' i j k :=
      (AlgebraicGeometry.pullbackSpecIso _ _ _).hom ≫
      Spec.map (CommRingCat.ofHom <| t' (ℱ i) (ℱ j) (ℱ k)) ≫
      (AlgebraicGeometry.pullbackSpecIso _ _ _).inv
  t_fac i j k := by
    ... -- after some simplifications
    exact t'_fac (ℱ i) (ℱ j) (ℱ k)
  cocycle i j k := by
    ... -- after some simplifications
    simpa using congr($(t'_cocycle (ℱ i) (ℱ j) (ℱ k)) x)
\end{lstlisting}
Hence, for any collection of good potion ingredients $\mathscr{F} = \left\{ S_i \middle| i \in \tau \right\}$, we have a scheme 
$\proj \mathscr{F}$.
\begin{lstlisting}
def Proj {τ : Type u} (ℱ : τ → GoodPotionIngredient 𝒜) : Scheme := 
  glueData ℱ |>.glued
\end{lstlisting}
\begin{flushright}
\cite[\lstinline|Proj| file location: {\footnotesize\textsf{Proj/Construction.lean}}]{projconstruction}
\end{flushright}
\end{construction}

\begin{landscape}
\thispagestyle{empty}
\begin{table}[h]
\refstepcounter{lstlisting} 
\caption{and \textbf{Listing~\thelstlisting} Mixing isomorphisms}
\addcontentsline{lol}{lstlisting}{\protect\numberline{\thelstlisting}Mixing isomorphisms}
\begin{tabular}{m{0.34\linewidth}p{0.45\linewidth}c}
  mathematical formulation & \Lean~formalsation  \\\hline
  $\potion(ST) \ox_{\potion(S)} \potion(SR) \cong A_{\sfs(T')} \ox_{\potion(S)} A_{\sfs(R')}$ &
  \vspace{-1em}
  \begin{lstlisting}
def mixingAux₀ :
    (S * T).Potion ⊗[S.Potion] (S * R).Potion ≃ₐ[S.Potion]
    (Localization T'.genSubmonoid) ⊗[S.Potion] 
    (Localization R'.genSubmonoid) :=
  Algebra.TensorProduct.congr
    (S.localizationAlgEquivPotion T').symm
    (S.localizationAlgEquivPotion R').symm
  \end{lstlisting}  \\
$A_{\sfs(T')} \ox_{\potion(S)} A_{\sfs(R')} \cong A_{\sfs(T')\sfs(R')}$ &
\vspace{-1em} 
   \begin{lstlisting}
def mixingAux₁ {R S T : GoodPotionIngredient 𝒜}
    (R' : PotionGen S.1 R.1) (T' : PotionGen S.1 T.1) :
    (Localization T'.genSubmonoid) ⊗[S.Potion] 
    (Localization R'.genSubmonoid) ≃ₐ[S.Potion]
    Localization (T'.genSubmonoid * R'.genSubmonoid) :=
  Localization.mulEquivTensor _ _ |>.symm
  \end{lstlisting} \\
$A_{\sfs(T')\sfs(R')} \cong A_{\sfs(T'\oplus R')}$&
\vspace{-2em}
  \begin{lstlisting}
def mixingAux₂ {R S T : GoodPotionIngredient 𝒜}
    (R' : PotionGen S.1 R.1) (T' : PotionGen S.1 T.1) :
    Localization (T'.genSubmonoid * R'.genSubmonoid) ≃ₐ[S.Potion]
    Localization (T'.disjUnion R').genSubmonoid :=
  Localization.equivEq 
    (PotionGen.disjUnion_genSubmonoid T' R').symm
  \end{lstlisting} \\
  $A_{\sfs(T'\oplus R')} \cong \potion(S(TR))$ &
\vspace{-2em}
  \begin{lstlisting}
def mixingAux₃ {R S T : GoodPotionIngredient 𝒜}
    (R' : PotionGen S.1 R.1) (T' : PotionGen S.1 T.1) :
    Localization (T'.disjUnion R').genSubmonoid ≃ₐ[S.Potion]
    (S * (T * R)).Potion :=
  S.localizationAlgEquivPotion (T'.disjUnion R')
  \end{lstlisting}  \\
$\potion(S(TR)) \cong \potion(RST)$ &
\vspace{-2em}
  \begin{lstlisting}
def mixingAux₄ (R S T : GoodPotionIngredient 𝒜) :
    (S * (T * R)).Potion ≃ₐ[S.Potion] (R * S * T).Potion :=
  AlgEquiv.ofRingEquiv (f := potionEquiv ...) ...
  \end{lstlisting}  \\
\end{tabular}
\label{potion-mixing-iso}
\end{table}
\restoregeometry
\end{landscape}

\begin{definition} \label{defi:full}
The (full) multi-graded Brenner--Schröer Proj scheme of $A$ is the scheme obtained by gluing the potion schemes associated with all good potion ingredients of $A$.
\begin{flushright}
\cite[\lstinline|Proj| file location: {\footnotesize\textsf{fullProj/Construction.lean}}]{projconstruction}
\end{flushright}
\end{definition} 

\begin{table}[h]
\refstepcounter{lstlisting} 
  \begin{lstlisting}
def fullProj : Scheme :=
  Proj (𝒜 := 𝒜) (fun P : GoodPotionIngredient 𝒜 => P)
  \end{lstlisting} 
\end{table}

\section{Enlarging Families of Good Potion Ingredients} \label{sec:enlarg-fam-potion}
In this section, we will show that given any family of good potion ingredients $\mathscr{F}$, we can replace it with a larger family $\mathscr{F'}$
while having $\proj \mathscr{F}$ and $\proj \mathscr{F'}$ being isomorphic as schemes. 
By enlarging the family of good potion ingredients, we will have more open sets of the form $\spec \potion(S)$ at our disposal.

  Suppose $\mathscr{F}$ and $\mathscr{F'}$ are two families of good potion ingredients indexed by $\tau$ and $\tau'$ respectively.
  We use the notation $\leanmath{le} : \mathscr{F} \le \mathscr{F'}$ to mean an injective function $\leanmath{le} : \tau \to \tau'$ such that $\mathscr{F}'\circ \leanmath{le} = \mathscr{F}$. We shall refer to such a map $\leanmath{le}$ as an {embedding} of families of good potion ingredients.
  
  \begin{lstlisting}[caption={Comparing two families of good potion ingredients}]
structure LE_ where
  (le : τ ↪ τ')
  (comp : ℱ' ∘ le = ℱ)
\end{lstlisting}

\begin{remark}{}{}
In~\Lean, the type class \lstinline|LE| is a proposition, our type \lstinline|LE_| contains data, 
hence writing $\mathscr{F} \le \mathscr{F'}$ without specifying the underlying function $\leanmath{le}$ is not accurate.
\end{remark}

\begin{construction}\label{projHomOfLE}
  Suppose $\leanmath{le}: \mathscr{F} \le \mathscr{F}'$ and $i \in \tau$, we have an isomorphism of rings 
  $\potion\left( \mathscr{F}'_{\leanmath{le}(i)} \right) \cong \potion\left( \mathscr{F}_i \right)$.
\begin{lstlisting}
def LE_.potionEquivMap (le : LE_ ℱ ℱ') (i : τ) : (ℱ' (le i)).Potion ≃+* (ℱ i).Potion :=
  potionEquiv (by simp)
\end{lstlisting}
Hence, we can glue a family of scheme morphisms
\[
\begin{tikzpicture}[remember picture,background rectangle/.style={fill=yellow!10}, show background rectangle]
\node {
\begin{tikzcd}[remember picture]
  \spec \potion\left( \mathscr{F}_i \right) \arrow[r, ""] & 
  \spec \potion\left( \mathscr{F}'_{\leanmath{le}(i)} \right) \arrow[r, "\iota"] &
  \proj \mathscr{F}' 
\end{tikzcd}
};
\end{tikzpicture}
\]
to form a morphism of schemes $\proj\leanmath{le}: \proj \scrF \to \proj \scrF'$,
\begin{lstlisting}
def projHomOfLE (le : LE_ ℱ ℱ') : Proj ℱ ⟶ Proj ℱ' :=
  Multicoequalizer.desc _ _
    (fun i ↦ (*@\sethlcolor{yellow!10}\hl{Spec.map (CommRingCat.ofHom <| le.potionEquivMap i) $\gg$ (glueData $\scrF$').$\iota$ (le i)}@*)) 
    ...
\end{lstlisting}
One can check that at every point $x \in \proj \scrF$, the morphism on stalks induced by $\proj \leanmath{le}$ is an isomorphism.
\begin{flushright}
  \cite[\lstinline|projHomOfLE| file location: {\footnotesize\textsf{Proj/OfLE.lean}}]{projconstruction}
\end{flushright}
\end{construction}

\begin{lemma}\label{projHomOfLE_base_injective}
  Topologically, $\proj \leanmath{le}$ is injective.
\begin{lstlisting}
lemma projHomOfLE_base_injective (le : LE_ ℱ ℱ') :
    Function.Injective (projHomOfLE le).base := ...
\end{lstlisting}
  \begin{flushright}
\cite[\lstinline|projHomOfLE_base_injective| file location: {\footnotesize\textsf{Proj/OfLE.lean}}]{projconstruction}
  \end{flushright}
\end{lemma}
\begin{proof}{}
  Suppose $x \in \spec \potion\left( \scrF_j \right)$ and $x' \in \spec \potion\left( \scrF_{j'} \right)$ are two points such that 
  $\proj \leanmath{le}(x) = \proj \leanmath{le}(x')$.
  Let us denote $X \in \spec \potion\left( \scrF'_{\leanmath{le}(j)} \right)$ and 
  $X' \in \spec \potion\left( \scrF'_{\leanmath{le}(j')} \right)$ to be the points corresponding to $x$ and $x'$ respectively.
  From $\proj \leanmath{le}(x) = \proj \leanmath{le}(x')$, we have the image of $X$ and $X'$ in $\proj \scrF'$ are the same.
  Therefore, either $\leanmath{le}(j) =\leanmath{le}(j')$ and $X = X'$ or there exists some point 
  $y \in \potion\left( \scrF'_{\leanmath{le}(j)} \scrF'_{\leanmath{le}(j')}\right)$ and $X = X' = y$ as points in $\proj \scrF'$.
  In either cases, we have $x = x'$.
\end{proof}

\begin{lemma}{}{}
  Let $U$ be an open set in $\proj \scrF$ and $i \in \tau$, $\proj \leanmath{le}\left( U \cap \potion\left( \scrF_i \right) \right)$ is an open set.
\begin{lstlisting}[caption={$\proj\leanmath{le}$ is an open map}]
abbrev interPotion (i : τ) : Opens (Proj ℱ) :=
  ((glueData ℱ).ι i).opensRange ⊓ U

lemma projHomOfLE_base_isOpenMap_aux (le : LE_ ℱ ℱ') (U : Opens (Proj ℱ)) (i : τ) :
    IsOpen <| (projHomOfLE le).base '' (interPotion U i) := by
\end{lstlisting}
\begin{flushright}
\cite[\lstinline|projHomOfLE_base_isOpenMap_aux| file location: {\footnotesize\textsf{Proj/OfLE.lean}}]{projconstruction}
\end{flushright}
\end{lemma}
\begin{proof}{}
  We consider the following diagram:
\[
\begin{tikzcd}[column sep = huge]
  U \cap \spec \potion(\scrF_i) \arrow[r] \arrow[d, equal]& \spec \potion(\scrF_i) \arrow[r, "\iota"] & \proj F \arrow[d, "\proj\leanmath{le}"]\\
  U \cap \spec \potion(\scrF_i) \arrow[r, "x"]& \spec \potion(\scrF'_{\leanmath{le}(i)}) \arrow[r, "\iota"] & \proj F' 
\end{tikzcd},
\]
where $x$ is defined as $\spec\left( \potion\left( \scrF'_{\leanmath{le}(i)} \right) \to \potion\left( \scrF_{i} \right) \right)$.
$\proj \leanmath{le}\left( U \cap \potion\left( \scrF_i \right) \right)$ is equal to the range of the second row; since the second row is a composition of open map, the range is open.
\end{proof}

\begin{corollary}{}{}
  The morphism $\proj \leanmath{le}: \proj \scrF \to \proj \scrF'$ is an open immersion.
\begin{lstlisting}[caption={$\proj\leanmath{le}$ is an open immersion}]
instance projHomOfLE_isOpenImmersion (le : LE_ ℱ ℱ') : IsOpenImmersion (projHomOfLE le) := ...
\end{lstlisting}
\begin{flushright}
\cite[\lstinline|projHomOfLE_isOpenImmersion| file location: {\footnotesize\textsf{Proj/OfLE.lean}}]{projconstruction}
\end{flushright}
\end{corollary}
\begin{proof}{}
  Since the induced stalk map at each point is an isomorphism, we only need to show that the morphism is a topological embedding.
  By~\cref{projHomOfLE_base_injective}, the morphism is injective; hence, we only need to show that the morphism is an open map.
  Let $U$ be an open set in $\proj \scrF$, then $U = \bigcup_i \left( U \cap \spec \potion(\scrF_i) \right)$.
  Thus, the image of $U$ is equal to the union of the images of $U \cap \spec \potion(\scrF_i)$ under $\proj \leanmath{le}$; 
  being a union of open sets, it is open as well.
\end{proof}

\begin{definition}\label{idealify}
  For any family $\scrF = \left\{ S_i \middle| i \in \tau \right\}$ of good potion ingredients, 
  we denote $\scrF^{+}$ to be the family indexed by the disjoint union of 
    $\tau$ and $\tau \times \left\{ S \middle| S~\text{is a good potion ingredient} \right\}$
  \[
  \begin{aligned}
    \scrF^{+} :  \tau \oplus \tau \times \left\{ S \middle| S~\text{is a good potion ingredient} \right\} &\to \left\{ S \middle| S~\text{is a good potion ingredient} \right\} \\
      i &\mapsto \scrF_i \\
      (i, S) &\mapsto \scrF_i S
  \end{aligned}.
  \]
  We define $\leanmath{le} : \scrF \le \scrF^{+}$ by the left inclusion.
  \begin{lstlisting}[caption={$\scrF^{+}$ and $\scrF \le \scrF^+$}, extendedchars=true]
abbrev idealify (ℱ : τ → GoodPotionIngredient 𝒜) :
    τ ⊕ (τ × GoodPotionIngredient 𝒜) → GoodPotionIngredient 𝒜 :=
  Sum.rec ℱ (fun p ↦ ℱ p.1 * p.2)

abbrev le_idealify (ℱ : τ → GoodPotionIngredient 𝒜) : LE_ ℱ (idealify ℱ) where
  le :=
  { toFun := Sum.inl
    inj' := Sum.inl_injective }
  comp := rfl
  \end{lstlisting}

\begin{flushright}
  \cite[\lstinline|idealify| file location: {\footnotesize\textsf{Proj/ofLE.lean}}]{projconstruction}
\end{flushright}
\end{definition}

\begin{theorem}\label{enlarge-family}
  The morphism $\proj \leanmath{le}: \proj \scrF \to \proj \scrF^{+}$ is an isomorphism.
  \begin{lstlisting}
instance proj_iso_proj_idealify : IsIso (projHomOfLE (le_idealify ℱ)) := 
  apply (config := { allowSynthFailures := true }) AlgebraicGeometry.IsOpenImmersion.to_iso
  rw [TopCat.epi_iff_surjective]
  ...
\end{lstlisting}
\begin{flushright}
\cite[\lstinline|proj_iso_proj_idealify| file location: {\footnotesize\textsf{Proj/OfLE.lean}}]{projconstruction}
\end{flushright}
\end{theorem}
\begin{proof}{}
  Since the morphism is an open immersion, we only need to show that topologically, it is surjective.
  Suppose $x \in \proj \scrF^{+}$, then either $x$ is in $\spec\potion(\scrF_i)$ for some $i \in \tau$ or 
  $x$ is in $\spec\potion(\scrF_i T)$ for some $i \in \tau$ and some good potion ingredient $T$.
  In the first case, the preimage of $x$ under $\proj \leanmath{le}$ is the point $x$ itself.

  In the second case, let $T'$ be a potion generator of $T$ over $\scrF_i$.  
   By~\cref{potion-loc-iso}, we have a ring isomorphism 
  $e : \potion(\scrF_i)_{\mathsf{s}(T')} \cong \potion\left( \scrF_i T \right)$. 
  Hence, the preimage of $x$ under $\proj \leanmath{le}$ is the pullback of $x$ along the morphism
\[
  \begin{tikzcd}
    \potion(\scrF_i) \arrow[r, "\frac{\bullet}{1}"] & \potion(\scrF_i)_{\mathsf{s}(T')} \arrow[r, "\cong"', "e"] & \potion\left( \scrF_i T \right)
  \end{tikzcd}.
\]
\end{proof}

\begin{remark}{}{}
$\scrF^+$ is the type theoretic version of $\scrF \cup \scrF \cdot \left\{ S\middle|S~\text{ is a good potion ingrient} \right\}$.
We have to use the construction in~\cref{idealify} because the generality we choose is families of good potion ingredients indexed by a type.
If $\scrF$ is a set of good potion ingredients, we see $\scrF$ as a family indexed by $\scrF$ itself.
If $\scrF \subseteq \scrF'$ are two sets of good potion ingredients,
we can define $\leanmath{le} : \scrF \le \scrF'$ by the usual inclusion.
Then, we can repeat the proof in~\cref{enlarge-family} and show that the morphism $\proj \leanmath{le}: \proj \scrF \to \proj \scrF^{+}$ is an isomorphism
where $\scrF^{+}$ is the set $\scrF \cup \scrF \cdot \left\{ S\middle|S~\text{ is a good potion ingrient} \right\}$.
\end{remark}

\section{Functoriality} \label{sec:functo}
In this section, we work with graded ring homomorphisms
$\Phi : A \to B$ and $\Psi : B \to C$ among $\iota$-graded $R_0$ algebras where $\iota$ is an abelian group.
\begin{lstlisting}
variable {τ ι R₀ A B C : Type u}
variable [AddCommGroup ι] [DecidableEq ι] [CommRing R₀]
variable [CommRing A] [Algebra R₀ A] {𝒜 : ι → Submodule R₀ A}
variable [GradedAlgebra 𝒜]
variable [CommRing B] [Algebra R₀ B] {ℬ : ι → Submodule R₀ B}
variable [GradedAlgebra ℬ]
variable [CommRing C] [Algebra R₀ C] {𝒞 : ι → Submodule R₀ C}
variable [GradedAlgebra 𝒞]

variable (Φ : 𝒜 →+* ℬ) (Ψ : ℬ →+* 𝒞)
\end{lstlisting}

\begin{theorem}{}{}
  Let $\scrF$ be a family of good potion ingredients indexed by $\tau$, we have a morphism of schemes 
  $\proj\Phi : \proj \Phi_{\star}\scrF \to \proj \scrF$,
  where $\Phi_{\star}\scrF$ is the family defined by $i \mapsto \Phi_{\star}\scrF_i$.
\begin{flushright}
\cite[\lstinline|Proj.map| file location: {\footnotesize\textsf{Proj/Functorial.lean}}]{projconstruction}
\end{flushright}
\end{theorem}
\begin{proof}{}
  For each $i \in \tau$, we have a ring homomorphism\footnote{see~\cref{sec:potions}} $\leanmath{potionToMap}: \potion(\scrF_i) \to \potion\left( \Phi_{\star}\scrF_i \right)$.
  Therefore, we have a family of morphisms of schemes 
  \[
\begin{tikzcd}
\spec{\potion\left( \Phi_{\star}\scrF_i \right)} \arrow[r] & \spec \potion\left( \scrF_i \right) \arrow[r, "\iota"] & \proj \scrF \\
\end{tikzcd}.
  \]
  We can glue these morphisms to form a morphism of schemes $\proj \Phi_{\star}\scrF \to \proj \scrF$.
\end{proof}

\begin{corollary}{}{}
  $\proj\id_A: \proj \scrF \to \proj \scrF$ is the identity and $\proj(\Psi \circ \Phi) = \proj \Phi \circ \proj \Psi$.
\begin{flushright}
\cite[\lstinline|Proj.map_id| and \lstinline|Proj.map_comp| file location: {\footnotesize\textsf{Proj/Functorial.lean}}]{projconstruction}
\end{flushright}
\end{corollary}

\begin{remark}{}{}
Since ${\id_A}_{\star}\scrF$ is not \emph{definitionally} equal to $\scrF$, \lstinline|Proj.map (.id 𝒜) ℱ  = 𝟙 _| will not type-check.
Hence, we use $\proj\leanmath{le}$ constructed in~\cref{projHomOfLE} to create a morphism $\proj {\id_A}_{\star}\scrF \to \proj \scrF$.
\begin{lstlisting}[caption={$\proj$ is a contravariant functor}]
lemma Proj.map_id (ℱ : τ → GoodPotionIngredient 𝒜) :
    Proj.map (.id 𝒜) ℱ  =
    projHomOfLE { le := { toFun := id, inj' _ _ h := h }, comp := ... } := by
  apply Multicoequalizer.hom_ext
  rintro i
  rfl
\end{lstlisting}
Similarly, since $(\Psi \circ \Phi)_{\star}\scrF$ is not definitionally equal to $\Psi_{\star}\left( \Phi_{\star}\scrF \right)$, we formalise 
$\proj(\Psi \circ \Phi) = \proj \Phi \circ \proj \Psi$ as follows:
\begin{lstlisting}
lemma Proj.map_comp (ℱ : τ → GoodPotionIngredient 𝒜) :
  Proj.map (Ψ.comp Φ) ℱ =
  projHomOfLE
  { le := { toFun := id, inj' _ _ h := h }, comp :=  ...] } ≫ 
  Proj.map Ψ _ ≫ Proj.map Φ ℱ := ...
\end{lstlisting}
In both cases, $\leanmath{le}$ is nothing but the identity function on $\tau$, but nevertheless, 
$\proj\leanmath{le}$ is not literally the identity map. Hence, set theoretically, we have checked that $\proj$ is a contravriant functor
from the category of graded ring to the category of schemes; however, type theoretically, we cannot write the functor down, at least not in the current formulation.
\end{remark}

\section{Dilatations of rings} \label{sec:dil-ring}
\begin{flushright}
\cite[file location: {\footnotesize\textsf{Dilatation}}]{projconstruction}
\end{flushright}

In this section, we formalize the notion of dilatations of rings as studied in \cite{stacks-project, M24, DMdS23}. 
We fix a commutative unital ring $A$. 
A multi-center in $A$ is a a collection $\{ [M_i, a_i]\}_{i \in I}$ where each $M_i$ is an ideal of $A$ and each $a_i$ is an element in $A$ \cite{M24}. 

\begin{lstlisting}
variable (A : Type*) [CommSemiring A]

structure Multicenter where
  (index : Type*)
  (ideal : index → Ideal A)
  (elem : index → A)
\end{lstlisting}

For each index, we define a larger ideal as $L_i= M_i + (a_i)$.

\begin{lstlisting}
def LargeIdeal (i : F.index) : Ideal A := F.ideal i + Ideal.span {F.elem i}
\end{lstlisting}

Let $\mathbb{N}_I $ be the monoid $\bigoplus _{i \in I} \mathbb{N} $. If $\nu = (\nu_1 , \ldots , \nu_i, \ldots ) \in \mathbb{N} _I $ we put $L^{\nu}= L_1 ^{\nu_1} \cdots L_i ^{\nu _i} \cdots  $ (product of ideals of $A$) and $a^{\nu}= a_1^{\nu_1 } \cdots a_i ^{\nu _i} \cdots$ (product of elements of $A$). Note that if $\nu \in \mathbb{N} _{I}$ is such that $\nu _i =0$ for all $i$, then $L_i^{\nu _i}= L^\nu =A$. We also put $a^{\mathbb{N} _I} = \{ a^\nu | \nu \in \mathbb{N} _I \} \subset A$.

\begin{definition} \label{generaldila} The dilatation of $A$ with multi-center $\{[M_i , a_i] \}_{i \in I}$ is the unital commutative ring $A[\big\{ \frac{M_i}{a_i}\big\}_{i \in I}]$ defined as follows:

$\bullet$  The underlying set of $A[\big\{ \frac{M_i}{a_i}\big\}_{i \in I}]$ is the set of equivalence classes of symbols $\frac{m}{a^{\nu}} $ where $ \nu \in \mathbb{N} _I$ and $m \in L^{\nu}$ under the equivalence relation 
\[ \frac{m}{a^\nu} \equiv \frac{p}{a^{\lambda}} \Leftrightarrow \exists \beta \in \mathbb{N} _I \text{ such that } ma^{\beta + \lambda }= p a^{\beta + \nu } \text{ in } A.\] From now on, we abuse notation and denote a class by any of its representative $\frac{m}{a^{\nu}}$ if no confusion is likely.

$\bullet$ The addition law is given by $\frac{m}{a^{\nu}}+ \frac{p}{a^{\beta}}= \frac{m a^{\beta } + p a^{\nu}}{a^{\beta + \nu}}$.

$\bullet$ The multiplication law is given by $\frac{m}{a^{\nu}}\times  \frac{p}{a^{\beta}} = \frac{ mp}{a^{\nu + \beta }}$.

$\bullet$ The additive neutral element is $\frac{0}{1}$ and the multiplicative neutral element is $\frac{1}{1}$.

We have a canonical morphism of rings $A \to A[\big\{ \frac{M_i}{a_i}\big\}_{i \in I}]$ given by $a \mapsto \frac{a}{1}$. We sometimes use the notations $A[\frac{M}{a}]$ or $A[\big\{ \frac{M_i}{a_i} :i \in I \big\}]$ to denote $A[\big\{ \frac{M_i}{a_i}\big\}_{i \in I}]$.
\end{definition}

In the formalization, we first introduce a structure encoding all symbolic fractions. We then introduce the equivalence class on fractions to define the dilatation as a set. 

\begin{lstlisting}


variable {A : Type*} [CommSemiring A]
variable {ι : Type*} (F : ι → Ideal A) (a : ι → A)


lemma familyPow_def (v : ι →₀ ℕ) : F^v = v.prod fun i k ↦ F i ^ k := rfl


structure PreDil where
  pow : F.index →₀ ℕ
  num : A
  num_mem : num ∈ F.LargeIdeal ^ pow

def r : F.PreDil → F.PreDil → Prop := fun x y =>
  ∃ β : F.index →₀ ℕ, x.num * F.elem^(β + y.pow) = y.num * F.elem^(β + x.pow)
  
  def setoid : Setoid (F.PreDil) where
  r := F.r
  iseqv :=
  { refl := r_refl
    symm {x y} := r_symm x y
    trans {x y z} := r_trans x y z }
    
    
variable (F) in
def Dilatation := Quotient F.setoid


scoped notation:max ring"["multicenter"]" => Dilatation (A := ring) multicenter


def mk (x : F.PreDil) : A[F] := Quotient.mk _ x
\end{lstlisting}

We then implement the addition and the multiplication on fractions and check compatibility with the equivalence relation. We thus obtain the formalization of the dilatation semiring. 

\begin{lstlisting}
instance instCommSemiring : CommSemiring A[F] 
\end{lstlisting}

We formalize that $A[F]$ is an $A$-algebra.

\begin{lstlisting}
def fromBaseRing : A →+* A[F] where
  toFun x := .mk
        { pow := 0
          num := x
          num_mem := by simp }
          
lemma algebraMap_eq : (algebraMap A A[F]) = fromBaseRing F
\end{lstlisting}

We prove that the image of $a^\nu$ in the dilatation is always a nonzerodivisor. 

\begin{lstlisting}
lemma nonzerodiv_image (v : F.index →₀ ℕ) : algebraMap A A[F] (F.elem^v) ∈ nonZeroDivisors A[F] 
\end{lstlisting}

As another fundamental property, we prove that $(a^\nu)= L^\nu$ in the dilatation. 

\begin{lstlisting}
lemma image_elem_LargeIdeal_equal  (v : F.index →₀ ℕ) :
 Ideal.span ({algebraMap A A[F] (F.elem^v)}) =
    Ideal.map (algebraMap A A[F]) (F.LargeIdeal^v)
\end{lstlisting}

We now implement that if $A$ is a ring (not only a semiring), then the dilatation is also a ring as it is compatible with opposite. 

\begin{lstlisting}
variable {A : Type _} [CommRing A] {F : Multicenter A}


instance : CommRing A[F] where

\end{lstlisting}

We now come to the universal property \cite{M24}. 

\begin{proposition} (Universal property) \label{univpropdilarings}
If $\chi : A \to B$ is a morphism of rings such that $\chi (a_i) $ is a non-zero-divisor and generates $\chi (L_i) B$ for all $i\in I$, then there exists a unique morphism $\chi '$ of $A$-algebras $A[\big\{ \frac{M_i}{a_i}\big\}_{i \in I}] \to B$. The morphism $\chi'$ sends $\frac{l}{a^\nu} $ to the unique element $b \in B $ such that $\chi ( a^\nu) b = \chi (l)$.
\end{proposition}
We split the formalization as follows. 
\begin{lstlisting}
lemma  lemma_exists_in_image [Algebra A B]
    (non_zero_divisor : ∀ i : F.index, (algebraMap A B) (F.elem i) ∈ nonZeroDivisors B)
    (gen : ∀ i, Ideal.span {(algebraMap A B) (F.elem i)} = Ideal.map (algebraMap A B) (F.LargeIdeal i)):
    (∀(ν : F.index →₀ ℕ) (m : F.LargeIdeal^ν) ,  (∃! bm : B ,  (algebraMap A B) (F.elem^ν) *bm=(algebraMap A B) (m) ))
    
    def def_unique_elem [Algebra A B] (v : F.index →₀ ℕ) (m : F.LargeIdeal^v)
    (non_zero_divisor : ∀ i : F.index, (algebraMap A B) (F.elem i) ∈ nonZeroDivisors B)
    (gen : ∀ i, Ideal.span {(algebraMap A B) (F.elem i)} = Ideal.map (algebraMap A B) (F.LargeIdeal i)): B :=
  (lemma_exists_in_image  F  non_zero_divisor gen v m).choose
  
  lemma def_unique_elem_unique  [Algebra A B] (v : F.index →₀ ℕ) (m : F.LargeIdeal^v)
    (non_zero_divisor : ∀ i : F.index, (algebraMap A B) (F.elem i) ∈ nonZeroDivisors B)
    (gen : ∀ i, Ideal.span {(algebraMap A B) (F.elem i)} = Ideal.map (algebraMap A B) (F.LargeIdeal i)):
    ∀ bm : B, (algebraMap A B) (F.elem^v) * bm = (algebraMap A B) m →  def_unique_elem F v m non_zero_divisor gen =bm
    
    def desc [Algebra A B]
    (non_zero_divisor : ∀ i : F.index, (algebraMap A B) (F.elem i) ∈ nonZeroDivisors B)
    (gen : ∀ i, Ideal.span {(algebraMap A B) (F.elem i)} = Ideal.map (algebraMap A B) (F.LargeIdeal i)) :
    A[F] →ₐ[A] B where
    
    lemma dsc_spec [Algebra A B] (v : F.index →₀ ℕ) (m : F.LargeIdeal^v)
    (non_zero_divisor : ∀ i : F.index, (algebraMap A B) (F.elem i) ∈ nonZeroDivisors B)
    (gen : ∀ i, Ideal.span {(algebraMap A B) (F.elem i)} = Ideal.map (algebraMap A B) (F.LargeIdeal i)):
    (algebraMap A B) (F.elem^v) * desc F non_zero_divisor gen (m/.v)  = (algebraMap A B) m
    
    lemma  lemma_exists_unique_morphism [Algebra A B]
    (non_zero_divisor : ∀ i : F.index, (algebraMap A B) (F.elem i) ∈ nonZeroDivisors B)
    (gen : ∀ i, Ideal.span {(algebraMap A B) (F.elem i)} = Ideal.map (algebraMap A B) (F.LargeIdeal i))
    (χ':A[F]→ₐ[A] B)  : χ' = desc F non_zero_divisor gen 
    
    lemma  lemma_exists_unique_morphism' [Algebra A B]
    (non_zero_divisor : ∀ i : F.index, (algebraMap A B) (F.elem i) ∈ nonZeroDivisors B)
    (gen : ∀ i, Ideal.span {(algebraMap A B) (F.elem i)} = Ideal.map (algebraMap A B) (F.LargeIdeal i))
    (χ χ' :A[F]→ₐ[A] B)  : χ = χ'
\end{lstlisting}


\begin{thebibliography}{99.}



\bibitem{BV25} A. Bouthier, E. Vasserot, {\it On the geometric Satake equivalence for Kac-Moody groups}, arXiv:2510.11466 (2025)

 
\bibitem{BS03} H.\,Brenner, S.\,Schröer, {\it Ample families, multihomogeneous spectra, and algebraization of formal schemes} Pacific J. Math. 208 (2003), no. 2, 209–230.

\bibitem{Bual} K. Buzzard et al., Schemes in Lean, Exp. Math. {\bf 31} (2022), no.~2, 355--363.
                 

\bibitem{DMdS23} A.~Dubouloz, A.~Mayeux, and J.~P. dos Santos, \emph{A survey on algebraic dilatations}, Field Institute Monographs, to appear.


\bibitem{Gr61} A.\,Grothendieck, {\it EGA : II. Etude globale elementaire de quelques classes de morphismes} Publications mathématiques de l'I.H.E.S., tome 8 (1961), p. 5-222


\bibitem{KSU21}
A.~Kuronya, P.~Souza, and M.~Ulirsch, \emph{Tropicalization of toric prevarieties},  preprint arXiv:2107.03139.

\bibitem{jjaassoonn_dimensiontheory} F. Li and J. Zhang, { \it Dimension Theory in Lean4}, 2025, https://github.com/jjaassoonn/DimensionTheory


\bibitem{mathlib4docs} The mathlib Community, {\it Mathlib4 Documentation}

\bibitem{M24}A. Mayeux \textit{Multi-centered Dilatations, Congruent Isomorphisms and Rost Double Deformation Space.} Transf. Groups. Volume 31, pages 1801–1850 (2026).

\bibitem{MR24}A. Mayeux and S. Riche, \textit{\it On multi-graded Proj schemes}, Publ. Res. Inst. Math. Sci. (2026) 

\bibitem{MZ} A. Mayeux and J. Zhang, \textit{The mechanization of science illustrated by the Lean formalization of the multi-graded Proj construction}, Springer Proc. Math. Stat., to appear.

\bibitem{projconstruction}A. Mayeux and J. Zhang, \textit{\it Multi-graded Proj construction in Lean4}, https://github.com/ProjConstruction/Proj (2025)


\bibitem{Mou21} L. de~Moura and S. Ullrich, The Lean 4 theorem prover and programming language, in {\it Automated deduction---CADE 28}, 625--635, Lecture Notes in Comput. Sci., 12699, Springer, Cham.


\bibitem{stacks-project} The stacks project.


\bibitem{Ta26} L. Tang, The $P^1$-motivic Gysin map, arXiv:2604.24888 (2026).

\bibitem{Z23}J. Zhang, \textit{Formalising the Proj construction in Lean}, in {\it 14th International Conference on Interactive Theorem Proving}, Art. No. 35, LIPIcs. Leibniz Int. Proc. Inform., 268 (2023).

\bibitem{WZ22}
E.\,Wieser and J.\,Zhang, {\it Graded rings in Lean's dependent type theory,} in {\it Intelligent computer mathematics}, 122--137, Lecture Notes in Comput. Sci. Lecture Notes in Artificial Intelligence, 13467 , Springer, Cham, 2022.

\bibitem{Ya26} T. Yasuda, \emph{An algorithm for the minimal model program in dimension three}, (2026) arXiv:2603.13703





\end{thebibliography}
\end{document}